\documentclass[10pt, a4paper]{article}
\usepackage[english]{babel}
\usepackage[latin1]{inputenc}
\usepackage[T1]{fontenc}
\usepackage{amsmath}
\usepackage{amsfonts}
\usepackage{amssymb}
\usepackage{amsthm}
\usepackage[affil-it]{authblk}
\usepackage{appendix}
\usepackage{enumitem}
\usepackage[normalem]{ulem}
\usepackage{dsfont}
\usepackage{float}
\usepackage{times}
\usepackage{tensor}
\usepackage{bbm}
\usepackage{braket}
\usepackage{array}
\usepackage[colorlinks,linkcolor=blue,citecolor=blue,urlcolor=blue]{hyperref}
\usepackage[capitalise]{cleveref}
\usepackage[left=4cm, right=4cm, bottom=2cm, top=2cm]{geometry}
\newtheorem{theorem}{Theorem}[section]
\newtheorem{lemma}[theorem]{Lemma}
\newtheorem{proposition}[theorem]{Proposition}

\newtheorem{assumption}[theorem]{Assumption}
\theoremstyle{definition}

\theoremstyle{remark}
\newtheorem*{remark}{Remark}
\newtheorem*{remarks}{Remarks}

\numberwithin{equation}{section}
\DeclareMathOperator{\diag}{diag}

\DeclareMathOperator{\trace}{tr}
\newcommand{\C}{\ensuremath{\mathbb{C}}}

\newcommand{\N}{\ensuremath{\mathbb{N}}}

\newcommand{\R}{\ensuremath{\mathbb{R}}}
\newcommand{\Z}{\ensuremath{\mathbb{Z}}}
\newcommand{\1}{\ensuremath{\mathds{1}}}

\newcommand{\Expec}[1]{\mathbb{E}\left[#1\right]}

\newcommand{\hilb}{\ensuremath{\mathcal{H}}}

\newcommand{\Tr}[1]{\ensuremath\trace\left(#1\right)}
\renewcommand{\Re}{\mathrm{Re}}
\renewcommand{\Im}{\mathrm{Im}}
\renewcommand{\bra}[1]{{\langle{#1}|}}
\makeatletter
\renewcommand{\ket}[1]{%
	\@ifnextchar\bra{\k@t{#1}\!}{\k@t{#1}}}
\newcommand{\k@t}[1]{{|{#1}\rangle}}
\makeatother
\renewcommand{\geq}{\geqslant}
\renewcommand{\leq}{\leqslant}

\begin{document}
\title{Bounds on the bipartite entanglement entropy for oscillator systems with or without disorder}
\author{Vincent Beaud, Julian Sieber, Simone Warzel}
\affil{Zentrum Mathematik, Boltzmannstr.~3, Technische Universit\"at M\"unchen}
\date{\today}
\maketitle

\begin{abstract}
	\noindent We give a direct alternative proof of an area law for the entanglement entropy of the ground state of disordered oscillator systems---a result due to Nachtergaele, Sims and Stolz \cite{Nachtergaele2013}. Instead of studying the logarithmic negativity, we invoke the explicit formula for the entanglement entropy of Gaussian states to derive the upper bound. We also contrast this area law in the disordered case with divergent lower bounds on the entanglement entropy of the ground state of one-dimensional ordered oscillator chains. 
	\vspace{0.2cm}
\end{abstract}

\section{Introduction}

Thanks to their relevance in the quantum information theory of optical systems, Gaussian quantum states (also known as quasi-free states in the mathematical physics literature) and the underlying systems of harmonic oscillators still enjoy widespread attention \cite{Braunstein2005,Braunstein2003,Cerf2007,Weedbrook2012,Lami2018,Derezinski2017,Adesso2014}. A popular measure for the entanglement structure of pure states is the bipartite entanglement entropy. It has been pointed out by Werner and Vidal \cite{Vidal2002} that this quantity is upper bounded by the logarithmic negativity of the quantum state. Ever since then many works have been devoted to the construction of bounds on the logarithmic negativity of Gaussian states \cite{Audenaert2002,Plenio2005,Cramer2006,Cramer2006a}, partially with the goal of explicitly confirming the general fact that ground states of gapped systems exhibit an area law bound on their entanglement entropy. More recently, Nachtergaele, Sims and Stolz  \cite{Nachtergaele2013} proved that a mobility gap induced by disorder also implies such an area law in oscillator systems---a fact which ought to hold more generally in disordered many-particle systems (see also~\cite{Eisert2010,Abdul-Rahman2018,Beaud2018,Fischbacher:2018rw} and references therein).

The present note mainly aims at demonstrating that the usual detour via the logarithmic negativity can be avoided if one is interested in effective bounds on the bipartite entanglement entropy of Gaussian quantum states. At first, we give an alternative proof of a result in \cite{Nachtergaele2013}, which involves a direct upper bound on the bipartite entanglement entropy of the ground state in systems of disordered harmonic oscillators. We then contrast this upper bound with lower bounds for one-dimensional ordered chains of oscillators where the excitation gap closes, thereby preventing an area law.

\subsection{Setting and assumptions}
\label{sec:model}

\noindent\textbf{The model.} We study coupled quantum oscillators with or without disorder on (finite) graphs. More precisely, let $\mathbb{G}=(\mathcal{V},\mathcal{E})$ be a graph with countable vertex set $\mathcal{V}$ and $\mathcal{E}$ a set of undirected edges. We will assume $\mathbb{G}$ to be of uniformly bounded degree, i.e. there exists $\mathcal{N}\in\N$ such that
\begin{align*}
	\sup_{x\in\mathcal{V}}|\{y\in\mathcal{V}\,|\,(x,y)\in\mathcal{E}\}|=\mathcal{N}<\infty.
\end{align*}
The system under consideration is given in terms of two real sequences $\lbrace h_{xy}^{(q)}\rbrace_{x,y\in\mathcal{V}}$ and $\lbrace h_{xy}^{(p)}\rbrace_{x,y\in\mathcal{V}}$. For any finite subset $\Lambda\subseteq\mathcal{V}$, the corresponding subsequences form two $|\Lambda|\times|\Lambda|$ square matrices that we denote by
\begin{align*}
	h_\Lambda^{(q)} = \lbrace h_{xy}^{(q)}\rbrace_{x,y\in\Lambda}\qquad\textnormal{and}\qquad
		h_\Lambda^{(p)} = \lbrace h_{xy}^{(p)}\rbrace_{x,y\in\Lambda}.
\end{align*}
A subset $\Lambda\subseteq\mathcal{V}$ is said to be connected whenever for any two sites $x,y\in\Lambda$ there exists a chain of vertices $\lbrace z_i\rbrace_{i=1}^{n+1}\subset\Lambda$, $n\in\N$, and edges $\lbrace(z_i,z_{i+1})\rbrace_{i=1}^n\subset\mathcal{E}$ connecting $z_1=x$ to $z_{n+1}=y$. For notational ease, we set
\begin{align*}
	\mathcal{L} = \lbrace \Lambda\subseteq\mathcal{V}\;|\;\Lambda \textnormal{ finite and connected}\rbrace.
\end{align*}
As in~\cite{Nachtergaele2013}, our analysis requires the following general assumptions, which are formulated to also accommodate disordered oscillators.
\begin{assumption}\label{ass:exhaustive}
	Let $ \lbrace h_{xy}^{(q)}\rbrace_{x,y\in\mathcal{V}}$ and 
	$\lbrace h_{xy}^{(p)}\rbrace_{x,y\in\mathcal{V}}$ be collections of real random variables on a common probability space $(\Omega,\mathcal{F},\mathbb{P})$. Assume that for any finite and connected subset $\Lambda\in\mathcal{L}$ the matrices $h_\Lambda^{(q)}, h_\Lambda^{(p)}\in\R^{|\Lambda|\times|\Lambda|}$ are symmetric, positive definite and satisfy the uniform norm bound
	\begin{align*}
		\sup\limits_{\Lambda\in\mathcal{L}}
		\max\left\{\|h_\Lambda^{(p)}\|,\|(h_\Lambda^{(p)})^{-1}\|,\|h_\Lambda^{(q)}\|\right\}
			\leq D\quad\textnormal{$\mathbb{P}$-a.s.}
	\end{align*}
	with some deterministic $D\in(0,\infty)$. Here, $\Vert\cdot\Vert$ denotes the operator norm.
\end{assumption}
\noindent We consider families of Hamiltonians $\lbrace H_\Lambda\rbrace_{\Lambda\in\mathcal{L}}$ of the form
\begin{align}\label{eq:finite_volume_uncorrelated}
	H_\Lambda
		=\sum\limits_{x,y\in\Lambda}h_{xy}^{(q)}q_x q_y+h_{xy}^{(p)}p_x p_y
		=\begin{pmatrix} q^T& p^T \end{pmatrix}\begin{pmatrix} h_\Lambda^{(q)}& 0\\	0& h_\Lambda^{(p)}\end{pmatrix}		
			\begin{pmatrix} q\\ p \end{pmatrix},
\end{align}
where $H_\Lambda$ acts on the Hilbert space
\begin{align}\label{eq:finite_volume_hilb}
	\mathcal{H}_\Lambda = \bigotimes\limits_{x\in\Lambda} L^2(\R,\mathrm{d}q_x)
\end{align}
and describes a coupled system of one-dimensional quantum oscillators, one of which sitting on each site of the subgraph $\Lambda$. Here we set $q=(q_x)_{x\in\Lambda}$ and $p=(p_x)_{x\in\Lambda}$, where $q_x$ denotes the position operator, i.e. the multiplication operator by $q_x$, while $p_x=-i\partial/\partial q_x$ stands for the momentum operator on $L^2(\R,\mathrm{d}q_x)$. These operators are self-adjoint on suitably chosen domains and satisfy the canonical commutation relations
\begin{align}\label{eq:commutation}
	[q_x,q_y]=[p_x,p_y]=0\quad\text{and}\quad [q_x,p_y]=i\delta_{xy},
\end{align}
where $\delta_{xy}=\lbrace 1\textnormal{ if }x=y,\textnormal{ and }0\textnormal{ otherwise}\rbrace$ denotes the Kronecker delta, cf.~\cite{Reed1975,Teschl2009}.
\bigskip\\
\noindent\textbf{Gaussian state and covariance matrix.} Since the Hamiltonian $H_\Lambda$~\eqref{eq:finite_volume_uncorrelated} is quadratic in $q_x$ and $p_x$, it may be diagonalized by a Bogolubov transformation. The computation is notably spelled out in~\cite{Nachtergaele2012} and involves the (positive definite) matrix
\begin{align}\label{eq:effective_single_particle}
		h_\Lambda = \left(h_\Lambda^{(p)}\right)^{1/2}h_\Lambda^{(q)}\left(h_\Lambda^{(p)}\right)^{1/2}
\end{align}
acting on $\R^{|\Lambda|}$. In particular, the ground state of $H_\Lambda$ is unique, pure and Gaussian (or quasi-free). For all (finite) $\Lambda$ the corresponding rank-one  density matrix will henceforth be denoted by $\rho_\Lambda$. Gaussian states are fully characterized (up to a unitary transformation) by their covariance matrix $\Gamma_\Lambda\in\R^{2|\Lambda|\times 2|\Lambda|}$ with entries
\begin{align}\label{eq:cov_matrix}
	\left(\Gamma_\Lambda\right)_{kl} = \Tr{\rho_\Lambda (r_k r_l + r_l r_k)},
\end{align}
where we used the shorthand $r=q\oplus p$, cf.~Appendix~\ref{app:gaussian_entropy}.
\bigskip\\
\noindent\textbf{Entanglement entropy and area law.} Let $\Lambda_0\subset\Lambda$ be a (finite) nontrivial subset of vertices and $\Lambda_0^c = \Lambda\setminus\Lambda_0$ its complement with respect to $\Lambda$. The entanglement structure of a pure state $\rho_\Lambda$ over the bipartition $\Lambda = \Lambda_0 \cup \Lambda_0^c$ is often quantified in terms of its bipartite \emph{entanglement (von Neumann) entropy}
\begin{align*}
	S(\rho_\Lambda ; \Lambda_0) := - \Tr{\rho_{\Lambda_0}\log\rho_{\Lambda_0}},
\end{align*}
where $\rho_{\Lambda_0}=\trace_{\Lambda_0^c}\left(\rho_\Lambda\right)$ is the reduced state of $\rho_\Lambda$ on $\Lambda_0$ and $\trace_{\Lambda_0^c}(\cdot)$ denotes the partial trace over the tensor component of the Hilbert space associated to $\Lambda_0^c$. As a consequence of the symmetry $S(\rho_\Lambda ; \Lambda_0) = S(\rho_\Lambda ; \Lambda_0^c)$, any lower or upper bound on the entanglement entropy shall in the sequel be understood as a bound on the minimum resp.~maximum over $\lbrace \Lambda_0,\Lambda_0^c\rbrace$.

The bipartite entanglement entropy of a generic multi-particle state is expected to grow linearly with the size of the subsystem $\Lambda_0$~\cite{Page1993,Foong1994}.  However, some states---such as isolated low-energy and localized states---depart from this regime and satisfy a so-called \emph{area law} where their entanglement entropy grows (to first order) with the cardinality of the boundary
\begin{align*}
	\partial\Lambda_0=\{x\in\Lambda_0\,|\,\exists y\in\Lambda\setminus\Lambda_0:(x,y)\in\mathcal{E}\}
\end{align*}
of the subsystem $\Lambda_0$.
\bigskip\\
\noindent\textbf{An example.} A simple instance of \eqref{eq:finite_volume_uncorrelated} consists of a lattice of quantum harmonic oscillators harmonically coupled to their neighbors by springs of constant strength. The corresponding Hamiltonian on a subset $\Lambda\subseteq\mathcal{V}$ reads
\begin{align}\label{eq:hamiltonian_example}
	H_\Lambda = \sum\limits_{x\in\Lambda}\left(\frac{1}{2m}p_x^2 + k_x q_x^2\right)
		+ \lambda\sum\limits_{\substack{x,y\in\Lambda\\(x,y)\in\mathcal{E}}} (q_x-q_y)^2,
\end{align}
where the mass $m$ of each oscillator and the coupling strength $\lambda$ of the interaction are positive constants. When $\lbrace k_x\rbrace_{x\in\Lambda}\subset\R_{\geq 0}$ are independent, identically distributed random variables, the one-particle operator $h_\Lambda$ is the Anderson model on $\ell^2(\Lambda)$. One readily convinces oneself that the family $\lbrace \overline{H}_\Lambda\rbrace_{\Lambda\in\mathcal{L}}$ of such Hamiltonians satisfies Assumption ~\ref{ass:exhaustive} whenever $\mathcal{V}$ is of bounded degree, the $\lbrace k_x\rbrace_{x\in\mathcal{V}}$ are almost surely uniformly bounded and $h_\Lambda^{(q)}>0$.

\subsection{Main results}
\label{sec:results}

Our first result consists in an alternative proof of the following theorem due to Nachtergaele, Sims and Stolz \cite{Nachtergaele2013}. Henceforth, $\langle\cdot,\cdot\rangle$ denotes the usual inner product on $\R^{|\Lambda|}$ and $\lbrace\delta_x\rbrace_{x\in\Lambda}\subset\R^{|\Lambda|}$ its canonical basis of vectors, with entries $\delta_x(y)=\delta_{xy}$.
\begin{theorem}[Theorem 2.2 in \cite{Nachtergaele2013}]\label{theorem:area_law}
	Let $\mathbb{G}=(\mathcal{V},\mathcal{E})$ be a graph of bounded degree $\mathcal{N}\in\N$ and $\lbrace H_\Lambda\rbrace_{\Lambda\in\mathcal{L}}$ a family of Hamiltonians of the form \eqref{eq:finite_volume_uncorrelated}, which satisfies Assumption \ref{ass:exhaustive}. Let $d(\cdot,\cdot)$ denotes the usual graph distance on the subgraph $\Lambda$ and assume furthermore that there exist $c<\infty$ and $\nu\in (2\log \mathcal{N},\infty)$ such that
	\begin{align}\label{eq:expectation_decay}
		\Expec{\left|\left\langle 
			(h_\Lambda^{(p)})^{1/2}\delta_x,h_\Lambda^{-1/2}(h_\Lambda^{(p)})^{1/2}\delta_y\right\rangle\right|}
			\leq c\, e^{-\nu d(x,y)}
	\end{align}
	for all finite connected subsets $\Lambda\in\mathcal{L}$ and all $x,y\in\Lambda$. Then, there exists $C\in(0,\infty)$ such that for any finite subset $\Lambda_0\subset\mathcal{V}$
	\begin{align}\label{eq:area_law}
		\Expec{S(\rho_\Lambda;\Lambda_0)} \leq C \left|\partial\Lambda_0\right|
	\end{align}
	for all $\Lambda\in\mathcal{L}$ with $\Lambda_0\subset\Lambda$.
\end{theorem}
Up to a possible improvement of the constant $C$ the content of this result coincides with that of~\cite{Nachtergaele2013}. The methods differ, though, and it is our main point to give a proof based on the explicit formula~\eqref{eq:bipartite_entanglement} for the entanglement entropy of general Gaussian states, rather than on their logarithmic negativity. The argument is spelled out in Section~\ref{sec:upper_bound} below. As the aforementioned formula~\eqref{eq:bipartite_entanglement} also applies to thermal states of oscillator systems, our general strategy yields a similar upper bound on their bipartite entanglement entropy (under a modified localization assumption as in~\cite[Theorem~2.3]{Nachtergaele2013}). Since the physical content of such a result is controversial, we refrain from engaging in this here.
\bigskip\\
\noindent The localization condition \eqref{eq:expectation_decay} is crucial to the validity of the above area law. To illustrate this, we consider the following explicit realization of the family of Hamiltonians \eqref{eq:hamiltonian_example} which is not subject to disorder and thus does not exhibit localization in its ground state:
\begin{equation}\label{eq:hamiltonian_explicit_real}
	\overline{H}_\Lambda =\, \sum\limits_{x\in\Lambda}p_x^2
		+ \sum\limits_{x=a}^{b-1} (q_{x+1}-q_x)^2 + q_a^2 +q_b^2 \, .
\end{equation}
It describes a one-dimensional chain of $b-a+1$ particles of mass $m=1/2$ connected by springs of constant strength $\lambda =1$ and pinned at its ends. Here $\Lambda\in\mathcal{L}$ is a finite and connected subset of $\mathcal{V}\in\lbrace\N,\Z\rbrace$, i.e. $\Lambda = [a,b]\cap\mathcal{V}$ for some $a,b\in\mathcal{V}$. As a second result, the bipartite entanglement entropy of the associated ground state is shown to grow with the size of the subsystem~$\Lambda_0$.
\begin{theorem}\label{theorem:entropy_grow}
	Let $\mathcal{V}\in\lbrace\N,\Z\rbrace$ with nearest-neighbor edges and $\lbrace\overline{H}_\Lambda\rbrace_{\Lambda\in\mathcal{L}}$ be the family of Hamiltonians \eqref{eq:hamiltonian_explicit_real} with ground state density matrix $\overline{\rho}_\Lambda$. Then there exist exhaustive sequences $\bigl\lbrace\Lambda_0^{(n)}\bigr\rbrace_{n\in\N}, \bigl\lbrace\Lambda^{(n)}\bigr\rbrace_{n\in\N}\subset\mathcal{L}$ with $\Lambda_0^{(n)}\subseteq\Lambda^{(n)}$ such that
	\begin{align}
		\lim\limits_{n\to\infty} S\left(\overline{\rho}_{\Lambda^{(n)}};\Lambda_0^{(n)}\right) = \infty.
	\end{align}
\end{theorem}
The pinning is responsible for the breaking of translation symmetry for finite $ \Lambda $, which in particular ensures the validity of $ h_\Lambda^{(q)} > 0 $ in Assumption~\ref{ass:exhaustive}. The divergence of the entanglement entropy may in both cases be traced to the closing of the spectral gap in the underlying one-particle Hamiltonian, cf.~\eqref{eq:orth_spec_decomp}. This behavior stands in contrast to the area law established in~\cite{Cramer2006} for the ground state of the periodic chain which is artificially modified so to exhibit a spectral gap above its ground state (see also~\cite{Cramer2006a}).

Lower bounds on the logarithmic negativity of ground states in quantum oscillator systems have been derived before~\cite{Audenaert2002}. The logarithmic negativity is however only an upper bound on the entanglement entropy.

\subsection{Outline of the method}
\label{sec:outline}

The main novel point of this note consists in demonstrating that the entanglement entropy may be estimated directly, without invoking the upper bound in terms of the logarithmic negativity. Our proofs rely on an explicit formula for the entanglement entropy of Gaussian states, which has the benefit to also yield lower bounds. We show how to obtain such lower bounds in the examples of~Theorem~\ref{theorem:entropy_grow}.
\bigskip\\
\noindent\textbf{Symplectic eigenvalues.} The commutation relations \eqref{eq:commutation} imply the Heisenberg matrix uncertainty inequality
\begin{align}\label{eq:heisenberg}
	\Gamma_\Lambda + i \Omega_\Lambda \geq 0,
\end{align}
where
\begin{align*}
	\Omega_\Lambda=\begin{pmatrix} 0&-\1\\	\1& 0	\end{pmatrix}\in\R^{2|\Lambda|\times 2|\Lambda|}
\end{align*}
defines a \emph{symplectic form} over $\R^{2|\Lambda|}$ through $(u,v)\mapsto \langle u,\Omega_\Lambda v\rangle$. The  \emph{symplectic group}
\begin{align*}
 	\mathsf{SP}(2|\Lambda|,\R)=\left\{S\in\R^{2|\Lambda|\times 2|\Lambda|}\,
 		\middle|\,S^T\Omega_\Lambda S=\Omega_\Lambda\right\}
\end{align*}
consists of all linear transformations $S$ of the vector $r=q\oplus p$ that conserve the commutation relations \eqref{eq:commutation}: $\left[(Sr)_k,(Sr)_l\right] = [r_k,r_l]=-i\left(\Omega_\Lambda\right)_{kl}$.

The covariance matrix associated to the ground state of a Hamiltonian of the form \eqref{eq:finite_volume_uncorrelated} admits the explicit expression
\begin{align}\label{eq:m}
	\Gamma_\Lambda=\begin{pmatrix}
		(h_\Lambda^{(p)})^{1/2}h_\Lambda^{-1/2}(h_\Lambda^{(p)})^{1/2}& 0\\
		0& (h_\Lambda^{(p)})^{-1/2}h_\Lambda^{1/2}(h_\Lambda^{(p)})^{-1/2}
		\end{pmatrix},
\end{align}
where $h_\Lambda$ is the one-particle operator defined in \eqref{eq:effective_single_particle}. By Assumption \ref{ass:exhaustive}, both $h_\Lambda$ and $\Gamma_\Lambda$ are symmetric and positive definite. This allows for the following spectral representation due to Williamson \cite{Williamson1936}.
\begin{proposition}[Proposition 3.2 in \cite{Nachtergaele2013}]\label{prop:williamson}
	Let $n\in\N$ and $\Gamma\in\R^{2n\times 2n}$ be symmetric and positive definite. Then there exists a symplectic matrix $S\in\mathsf{SP}(2n,\R)$ such that
	\begin{align}\label{eq:diag_williamson}
		S^T\Gamma S=\begin{pmatrix} \mathcal{G}& 0\\	0& \mathcal{G}	\end{pmatrix},
	\end{align}
	where $\mathcal{G}=\diag(\gamma_1,\dots,\gamma_n)>0$. The symplectic eigenvalues $\sigma_\mathrm{symp}(\Gamma):=      \{\gamma_k\}_{k=1}^n$ can be computed as the positive eigenvalues of $i\Gamma^{1/2}\Omega_n \Gamma^{1/2}$ or as the imaginary part of the eigenvalues of $\Gamma\Omega_n$. Furthermore,
	\begin{align}\label{eq:heisenberg_matrix}
	\Gamma+i\Omega_n\geq 0
	\end{align}
	if and only if $\gamma_k\geq 1$ for all $k=1,\dots,n$.
\end{proposition}
\begin{remarks}
	1. Let us emphasize that, in the above definition, the symplectic spectrum of $\Gamma$ only comprises the eigenvalues of $\mathcal{G}$, though with multiplicity. To avoid redundancy, the fact that $\mathcal{G}$ appears twice in the symplectic diagonalization of $\Gamma$ remains unrecorded.
	
	\noindent 2. Gaussian functionals are in general characterized by a positive semidefinite covariance matrix, cf. Appendix~\ref{app:gaussian_entropy}. However, they fail to be states whenever \eqref{eq:heisenberg} is not satisfied. The uncertainty relation \eqref{eq:heisenberg} implies that the covariance matrix is positive definite.
\end{remarks}
\begin{lemma}\label{corollary:symplectic}
	Let $A,B\in\R^{n\times n}$ be symmetric, positive definite matrices and
	\begin{align}\label{eq:block_matrix}
		\Gamma = \begin{pmatrix} A & 0\\ 0 & B	\end{pmatrix}\in\R^{2n\times 2n}.
	\end{align}
	Then, $\sigma_\mathrm{symp}(\Gamma)$ consists of the positive square roots of the eigenvalues of $AB$, $A^{1/2}BA^{1/2}$ or $B^{1/2}AB^{1/2}$, counted with multiplicity.
\end{lemma}
\begin{proof}
	The matrices $AB$, $A^{1/2}BA^{1/2}$ and $B^{1/2}AB^{1/2}$ being similar, their spectra coincide and by assumption only consist of (strictly) positive eigenvalues. By Proposition \ref{prop:williamson} it now suffices to show that the positive eigenvalues of $i\Gamma^{1/2}\Omega_n\Gamma^{1/2}$ agree with the square root of the eigenvalues of $A^{1/2}BA^{1/2}$, all counted with multiplicity.
	
	In fact, since $\Gamma^{1/2} = A^{1/2}\oplus B^{1/2}$ is symmetric and $\Omega_n$ antisymmetric, the matrix $\Gamma^{1/2}\Omega_n\Gamma^{1/2}$ is antisymmetric. Its eigenvalues are thus grouped in pairs $\pm i\gamma_k$ with $\gamma_k\geq 0$ for all $k=1,\ldots,n$. Accordingly the spectrum of $i\Gamma^{1/2}\Omega_n\Gamma^{1/2}$ with multiplicity reads $\left\lbrace\pm\gamma_k\right\rbrace_{k=1}^n$. Observing finally that
	\begin{align*}
		\left(i\Gamma^{1/2}\Omega_n\Gamma^{1/2}\right)^2
			= \begin{pmatrix} A^{1/2}BA^{1/2} & 0 \\ 0&B^{1/2}AB^{1/2}\end{pmatrix}
	\end{align*}
	where $A^{1/2}BA^{1/2}$ and $B^{1/2}AB^{1/2}$ are similar, one concludes that all $\gamma_k>0$ and that they coincide with the positive square root of the eigenvalues of $A^{1/2}BA^{1/2}$.
\end{proof}
\begin{remark}
	As a direct consequence, the symplectic eigenvalues of $\Gamma_\Lambda$ in \eqref{eq:m} are all $1$, since the diagonal blocks are inverse of each other.
\end{remark}

Gaussian states remain Gaussian under partial traces, with covariance matrix truncated correspondingly. 
\begin{proposition}\label{prop:reduction_Gaussian}
	Let $\Lambda_0\subset\Lambda$ be finite subsets of $\mathcal{V}$ and $\rho_\Lambda$ the density matrix of a Gaussian state on $\mathcal{H}_\Lambda$ with covariance matrix $\Gamma_\Lambda$. Then, the reduced state on $\mathcal{H}_{\Lambda_0}$ with density matrix $\rho_{\Lambda_0}=\trace_{\Lambda_0^c}\left(\rho_\Lambda\right)$ is Gaussian with covariance matrix 
	\begin{align}\label{eq:restricted_cov_mat}
		\Gamma_{\Lambda_0} 
			= \left(\iota_{\Lambda_0}^{*}\oplus\iota_{\Lambda_0}^{*}\right)
				\Gamma_\Lambda\left(\iota_{\Lambda_0}\oplus\iota_{\Lambda_0}\right)
				\:\in\R^{2|\Lambda_0|\times 2|\Lambda_0|}
	\end{align}
	given in terms of the canonical embedding $\iota_{\Lambda_0}: \R^{|\Lambda_0|}\hookrightarrow\R^{|\Lambda|}$. The uncertainty relation~\eqref{eq:heisenberg} holds with the truncated symplectic form $\Omega_{\Lambda_0} = (\iota_{\Lambda_0}^{*}\oplus\iota_{\Lambda_0}^{*})\Omega_\Lambda\left(\iota_{\Lambda_0}\oplus\iota_{\Lambda_0}\right)$.
\end{proposition}
A proof can be found in Appendix~\ref{app:gaussian_entropy}. Notice the following notational rule: quantities inherent to the Hamiltonian $H_\Lambda$ on the Hilbert space $\hilb_\Lambda$ feature the subscript $(\cdot)_\Lambda$, whereas restricted or reduced quantities onto $\Lambda_0$ exhibit the subscript $(\cdot)_{\Lambda_0}$. The latter should in particular not be confused with quantities inherent to the Hamiltonian $H_{\Lambda_0}$ on $\hilb_{\Lambda_0}$, which will never appear in this work. In particular, the symplectic eigenvalues of a restricted covariance matrix $\Gamma_{\Lambda_0}$ generally differ from $1$, even though the underlying unrestricted $\Gamma_\Lambda$ may be of the form \eqref{eq:m}.
\bigskip\\
\noindent\textbf{Explicit formula for the entanglement entropy.} Our results rest on the following explicit expression relating the von Neumann entropy of a Gaussian state to the symplectic spectrum of its covariance matrix. Thanks to Proposition \ref{prop:reduction_Gaussian}, this result also applies to the bipartite entanglement entropy of any Gaussian state.
\begin{proposition}\label{prop:entropy_formula}
	Let $\Lambda\subset\mathcal{V}$ be a finite subset and $\rho_\Lambda$ the density matrix of a Gaussian state with (positive definite) covariance matrix $\Gamma_\Lambda$. The von Neumann entropy of $\rho_\Lambda$ satisfies
	\begin{align}\label{eq:von_neumann_entropy}
		S\left(\rho_\Lambda\right) = - \Tr{\rho_\Lambda \log\rho_\Lambda}
			= \sum\limits_{\gamma\in\sigma_\mathrm{symp}(\Gamma_\Lambda)} f(\gamma)
	\end{align}
	where $\sigma_\mathrm{symp}(\cdot)$ denotes the symplectic spectrum with multiplicity (see Proposition~\ref{prop:williamson}) and
	\begin{align}\label{eq:entropy_function}
		f(x) = \frac{x+1}{2}\log\left(\frac{x+1}{2}\right) - \frac{x-1}{2}\log\left(\frac{x-1}{2}\right)
	\end{align}
	for all $x\in(1,\infty)$ and $f(1)=0$. In particular, the bipartite entanglement entropy of $\rho_\Lambda$ is given by
	\begin{align}\label{eq:bipartite_entanglement}
		S\left(\rho_\Lambda;\Lambda_0\right) = \sum\limits_{\gamma\in\sigma_\mathrm{symp}(\Gamma_{\Lambda_0})} f(\gamma)
	\end{align}
	for any bipartition $\Lambda=\Lambda_0\cup\Lambda_0^c$.
\end{proposition}
This statement seems to date back to at least~\cite{Holevo1999}, albeit without proof. For the convenience of the reader, Appendix~\ref{app:gaussian_entropy} contains a mathematical proof. Let us emphasize that, by~\eqref{eq:heisenberg} and Proposition~\ref{prop:williamson}, the symplectic eigenvalues of $\Gamma_{\Lambda}$ (and $\Gamma_{\Lambda_0}$) are all greater or equal to $1$, and the above formulae are thus well defined. Theorems~\ref{theorem:area_law} and \ref{theorem:entropy_grow} now rest on suitable estimates of the entropy function $f$ and bounds on the symplectic eigenvalues.
\begin{lemma}\label{lem:function}
	The function $f: [1,\infty) \to \R $ defined in \eqref{eq:entropy_function} satisfies:
 	\begin{enumerate}[label=(\roman*)]
 		\item $f$ is continuous, strictly monotone increasing and concave with $\lim_{x\to 1} f'(x) = \infty$ \label{it:increasing}.
 		\item there exists $C\in(0,1]$ such that $ f(x)\leq C \sqrt{x^2-1}$, for all $x\in[1,\infty)$ \label{it:upper_bound}.
 		\item $f(x)\geq \log(x)$ for all $x\in[1,\infty)$.\label{it:lower_bound}
 	\end{enumerate}
\end{lemma}
\begin{proof}
	The continuity in \ref{it:increasing} is immediate as $\lim_{x\to 0}x\log(x) = 0$. The remaining assertions follow from
	\begin{align*}
		f'(x) = \frac{1}{2}\log\left(\frac{x+1}{x-1}\right) > 0,
			\qquad f''(x) = -\frac{1}{x^2-1} <0.
	\end{align*}	
	We show \ref{it:upper_bound} for $C=1$. Let $g(x) = \sqrt{x^2-1}$. Then $g(1)=0$ and for all $x\in(1,\infty)$ we have
	\begin{align*}
		2g'(x) = \frac{2x}{\sqrt{x^2-1}} \geq \frac{x+1}{\sqrt{x^2-1}}
			\geq 2\log\left(\frac{x+1}{\sqrt{x^2-1}}\right)
			= \log\left(\frac{x+1}{x-1}\right) = 2f'(x),
	\end{align*}
	where in the second inequality we used that $y\geq 2\log(y)$ for all $y\in(0,\infty)$. Similarly, \ref{it:lower_bound} follows from $\log(1) = 0$ and from the inequality $\log(y)\geq 2\tfrac{y-1}{y+1}$ for all $y\in[1,\infty)$, which implies
	\begin{align*}
		f'(x) = \frac{1}{2}\log\left(\frac{x+1}{x-1}\right)
			\geq \left.\frac{y-1}{y+1}\right\vert_{y=\tfrac{x+1}{x-1}} = \frac{1}{x} = \log'(x)
	\end{align*}
	for all $x\in(1,\infty)$.
\end{proof}
\begin{remarks}
	1. One can show that for each $\alpha\in(0,1)$ there exists a constant $C_\alpha\in(0,\infty)$ such that $f(x)\leq C_\alpha (x^2-1)^\alpha$. By \ref{it:increasing} this constant blows up in the limit $\alpha\to\lbrace 0,1\rbrace$.
	
	\noindent 2. While the coefficient of the lower bound \ref{it:lower_bound} cannot be improved, a more careful analysis of \ref{it:upper_bound} yields the optimal constant	$C=\sqrt{x_0^2-1}(\log(2)-\log(\sqrt{x_0^2-1}))\simeq 0.56447$,	where $x_0 \simeq 1.6367$ is the unique solution different from $1$ of $x\log\bigl(\frac{x^2-1}{4}\bigr) = \log\bigl(\frac{x-1}{x+1}\bigr)$.
\end{remarks}

\section{Upper bound and area law with disorder}
\label{sec:upper_bound}

The following bound relates the entanglement entropy of ground states of quite general oscillator systems to decay properties of the underlying single-particle Hamiltonian.
\begin{lemma}
	Let $H_\Lambda$ be a Hamiltonian of the form \eqref{eq:finite_volume_uncorrelated}, which satisfies Assumption \ref{ass:exhaustive}, and $\rho_\Lambda$ the density matrix associated to its (unique) ground state. For any $\Lambda_0\subset\Lambda$ we have
	\begin{equation}
		S\left(\rho_\Lambda;\Lambda_0\right)
			\leq D^{1/2} \sum\limits_{\substack{x\in\Lambda_0\\y\in\Lambda\setminus\Lambda_0}}
				\vert\langle\delta_x,
				(h_\Lambda^{(p)})^{1/2}h_\Lambda^{-1/2}(h_\Lambda^{(p)})^{1/2}\delta_y\rangle\vert^{1/2}
				\qquad\mathbb{P}\textnormal{-a.s.}, \label{eq:bound_mat_element}
	\end{equation}
	where $D$ denotes the constant from Assumption \ref{ass:exhaustive}.
\end{lemma}
\begin{proof}
	The ground state of $H_\Lambda$ is Gaussian and thus fully characterized by the covariance matrix $\Gamma_\Lambda$ explicitly given here in \eqref{eq:m} by
	\begin{align*}
	\Gamma_\Lambda = \begin{pmatrix} A_\Lambda & 0 \\ 0 & A_\Lambda^{-1}\end{pmatrix}
		\quad\textnormal{with}\quad A_\Lambda = (h_\Lambda^{(p)})^{1/2}h_\Lambda^{-1/2}(h_\Lambda^{(p)})^{1/2}.
	\end{align*}
	By Proposition~\ref{prop:reduction_Gaussian}, the reduced state on $\Lambda_0$ is Gaussian, with truncated covariance matrix
	\begin{align}\label{eq:restricted_cov_mat_2}
		\Gamma_{\Lambda_0} = \begin{pmatrix} \iota_{\Lambda_0}^{*}A_\Lambda\iota_{\Lambda_0} & 0 \\
			 0 & \iota_{\Lambda_0}^{*}A_\Lambda^{-1}\iota_{\Lambda_0}\end{pmatrix},
	\end{align}
	where $\iota_{\Lambda_0}: \R^{|\Lambda_0|}\hookrightarrow\R^{|\Lambda|}$ denotes the canonical embedding. According to Proposition~\ref{prop:entropy_formula} and Lemma~\ref{lem:function}, the entanglement entropy of the ground state over the bipartition $\Lambda = \Lambda_0\cup\Lambda_0^c$ admits the upper bound
	\begin{align}\label{eq:bipartite_entanglement_2}
		S\left(\rho_\Lambda;\Lambda_0\right) = \sum\limits_{\gamma\in\sigma_\mathrm{symp}(\Gamma_{\Lambda_0})} f(\gamma)
			\:\leq \sum\limits_{\gamma\in\sigma_\mathrm{symp}(\Gamma_{\Lambda_0})} \sqrt{\gamma^2 - 1}.
	\end{align}
	Applying Lemma~\ref{corollary:symplectic} to the matrix~\eqref{eq:restricted_cov_mat_2}, the symplectic eigenvalues $\gamma\in\sigma_\mathrm{symp}(\Gamma_{\Lambda_0})$ are the square roots of the (positive) eigenvalues of
	\begin{align*}
		\iota_{\Lambda_0}^{*} A_\Lambda \iota_{\Lambda_0} \iota_{\Lambda_0}^{*} A_\Lambda^{-1} \iota_{\Lambda_0}
			&= \iota_{\Lambda_0}^{*} A_\Lambda P_{\Lambda_0} A_\Lambda^{-1} \iota_{\Lambda_0} 
			= \mathds{1}_{\Lambda_0} + \iota_{\Lambda_0}^{*}[A_\Lambda,P_{\Lambda_0}]A_\Lambda^{-1}\iota_{\Lambda_0},
	\end{align*}
	where $\iota_{\Lambda_0}\iota_{\Lambda_0}^{*} = P_{\Lambda_0}$ denotes the orthogonal projection onto $\R^{|\Lambda_0|}$ and $\iota_{\Lambda_0}^{*}\iota_{\Lambda_0} = \mathds{1}_{\Lambda_0}$ is the identity on $\R^{|\Lambda_0|}$. Recall that $\Gamma_{\Lambda_0}$ fulfills the Heisenberg uncertainty relation \eqref{eq:heisenberg}. By Proposition \ref{prop:williamson}, its symplectic spectrum is thus contained in $[1,\infty)$, whence the eigenvalues of $\iota_{\Lambda_0}^{*}[A_\Lambda,P_{\Lambda_0}]A_\Lambda^{-1}\iota_{\Lambda_0}$ are all nonnegative. Inserting in the right-hand side of~\eqref{eq:bipartite_entanglement_2} yields
	\begin{align*}
		\sum\limits_{\gamma\in\sigma_\mathrm{symp}(\Gamma_{\Lambda_0})} \sqrt{\gamma^2 - 1}
			&\leq \Tr{\vert \iota_{\Lambda_0}^{*}[A_\Lambda,P_{\Lambda_0}]A_\Lambda^{-1}\iota_{\Lambda_0} \vert^{1/2}}\\
			&\leq D^{1/2} \Vert P_{\Lambda_0} A_\Lambda P_{\Lambda_0}^\perp\Vert_{1/2}^{1/2}
			\qquad\textnormal{$\mathbb{P}$-a.s.},
	\end{align*}
	where $P_{\Lambda_0}^\perp = \mathds{1}_\Lambda - P_{\Lambda_0}$ and $\Vert\cdot\Vert_{1/2}$ with $\Vert O \Vert_{1/2}^{1/2} = \Tr{\vert O\vert^{1/2}}$ denotes the Schatten $1/2$-quasinorm. In the last step we used that $\Vert O_1 O_2\Vert_{1/2} \leq \Vert O_1\Vert_{1/2} \Vert O_2\Vert$ for any two operators $O_1,O_2$, followed by the $\mathbb{P}$-almost sure bound $\Vert A_\Lambda^{-1}\Vert\leq D$ from Assumption~\ref{ass:exhaustive}, as well as
$\iota_{\Lambda_0}^{*}[A_\Lambda,P_{\Lambda_0}] =\iota_{\Lambda_0}^{*}P_{\Lambda_0}[A_\Lambda,P_{\Lambda_0}]= -\iota_{\Lambda_0}^{*} P_{\Lambda_0} A_\Lambda P_{\Lambda_0}^\perp$.
	The claim then follows from
	\begin{align*}
		\Vert P_{\Lambda_0} A_\Lambda P_{\Lambda_0}^\perp\Vert_{1/2}^{1/2}
			&\leq \sum\limits_{x,y\in\Lambda}
				\vert\langle\delta_x,P_{\Lambda_0} A_\Lambda P_{\Lambda_0}^\perp\delta_y\rangle\vert^{1/2}
			= \sum\limits_{\substack{x\in\Lambda_0\\y\in\Lambda\setminus\Lambda_0}}
				\vert\langle\delta_x,A_\Lambda\delta_y\rangle\vert^{1/2},
	\end{align*}
	where the inequality derives from a more general Schatten quasinorm estimate proven hereafter. Let $A=\left(A(j,k)\right)_{j,k=1}^n\in\C^{n\times n}$ and $\alpha\in(0,1]$, the Jensen--Peierls inequality yields
	\begin{align*}
		\|A\|_\alpha^\alpha = \Tr{|A|^\alpha} = \sum\limits_{k=1}^n |A|^{\alpha}(k,k)
			\leq \sum\limits_{k=1}^n |A|(k,k)^\alpha .
	\end{align*}
	Using the polar decomposition $A = U|A|$ for a suitable unitary $U$, we conclude through the Cauchy--Schwarz inequality that
	\begin{align*}
		\|A\|_\alpha^\alpha &\leq \sum\limits_{k=1}^n \left[(U^{*}A)(k,k)\right]^\alpha
			\leq \sum\limits_{k=1}^n \Vert U^{*}A(\cdot,k)\Vert^\alpha
			= \sum\limits_{k=1}^n \Vert A(\cdot,k)\Vert^\alpha
			\leq \sum\limits_{j,k=1}^{n}|A(j,k)|^\alpha.
	\end{align*}
	Here, $A(\cdot,k)$ denotes the $k$-th column of $A$, $\|\cdot\|$ is the Euclidean norm and the last inequality is by $(a+b)^\alpha \leq a^\alpha + b^\alpha$ for all $a,b\geq 0$.
	This concludes the proof.
\end{proof}
We may now present the short proof of Theorem~\ref{theorem:area_law}.
\begin{proof}[Proof of Theorem~\ref{theorem:area_law}]
	By combining \eqref{eq:bound_mat_element} with Jensen's inequality and the localization condition \eqref{eq:expectation_decay}, we obtain
	\begin{align*}
		\Expec{S(\rho_\Lambda;\Lambda_0)} 
			&\leq D^{1/2} \sum\limits_{x\in\Lambda_0}\sum\limits_{y\in\Lambda\setminus\Lambda_0}
				\Expec{\vert\langle\delta_x,
				(h_\Lambda^{(p)})^{1/2}h_\Lambda^{-1/2}(h_\Lambda^{(p)})^{1/2}\delta_y\rangle\vert}^{1/2}\\
			&\leq (c\,D)^{1/2} \sum\limits_{x\in\Lambda_0}
				\sum\limits_{y\in\Lambda\setminus\Lambda_0}
				e^{-\frac{\nu}{2}d(x,y)}.
	\end{align*}
	Now, since the underlying graph $\mathbb{G}=(\mathcal{V},\mathcal{E})$ is of degree bounded by $\mathcal{N}$, the following sum is finite for any $\mu\in(\log \mathcal{N},\infty)$:
\begin{align*}
	\sup\limits_{x\in\mathcal{V}}\sum\limits_{y\in\mathcal{V}} e^{-\mu d(x,y)}
		=:D_\mu<\infty.
\end{align*}
	By assumption, this holds in particular for $\mu=\nu/2$. Hence, for any $\Lambda_0\subset\Lambda\subset\mathcal{V}$ we have
	\begin{align*}
		\sum_{x\in\Lambda_0}^{}\sum_{y\in\Lambda\setminus\Lambda_{0}}e^{-\frac{\nu}{2} d(x,y)}
			\leq (D_{\nu/2})^2|\partial\Lambda_{0}|,
	\end{align*}
cf.~Lemma 4.2 in \cite{Nachtergaele2013}, yielding the claimed area law.
\end{proof}

\section{Instance of entropy growth without disorder}
\label{sec:lower_bound}

To contrast with the established area law in disordered oscillators systems, we consider the family of ordered Hamiltonians $\lbrace\overline{H}_\Lambda\rbrace_{\Lambda\in\mathcal{L}}$ \eqref{eq:hamiltonian_explicit_real} on the one-dimensional lattice $\mathcal{V}\in\lbrace\N,\Z\rbrace$. As stated in Theorem~\ref{theorem:entropy_grow} and proved in this section, the bipartite entanglement entropy of the ground state diverges with the size of the subsystem $\Lambda_0\subset\Lambda$, rather than with the size of its boundary. The Hamiltonian $\overline{H}_\Lambda$ is expressed as a quadratic form
\begin{equation}\label{eq:hamiltonian_explicit_real2}
	\overline{H}_\Lambda = \begin{pmatrix} q^T& p^T \end{pmatrix}\begin{pmatrix} \overline{h}_\Lambda^{(q)}& 0\\
			0& \overline{h}_\Lambda^{(p)}\end{pmatrix}		
			\begin{pmatrix} q\\ p \end{pmatrix}
\end{equation}
in terms of the real sequences
\begin{align}\label{eq:sequences_explicit_real}
	\overline{h}^{(p)}_{xy} = \delta_{xy} \qquad\textnormal{and}\qquad
		\overline{h}^{(q)}_{xy} = \begin{cases} 2, & \textnormal{for }x=y,\\ -1, &\textnormal{for }|x-y|=1,\\
		0, &\textnormal{otherwise},\end{cases}
\end{align}
along with the prescription $\overline{h}^{(\sharp)}_\Lambda = \left\lbrace \overline{h}^{(\sharp)}_{xy} \right\rbrace_{x,y\in\Lambda}$ for $\sharp\in\lbrace p,q\rbrace$.

Since for any $\Lambda\in\mathcal{L}$ the matrices $\overline{h}_\Lambda^{(p)},\overline{h}_\Lambda^{(q)}\in\R^{|\Lambda|\times|\Lambda|}$ are the identity, respectively the negative discrete Laplacian on $\ell^2(\Lambda)$, the spectral properties of $\overline{H}_\Lambda$ are easily obtained and collected in Subsection~\ref{subsec:properties_hamilt}. The proof of Theorem~\ref{theorem:entropy_grow} then proceeds by distinguishing the cases $\mathcal{V}=\Z$ and $\mathcal{V}=\N$ in Subsections~\ref{subsec:proof_Z} and \ref{subsec:proof_N}. This distinction is due to the presence or absence of translation invariance in the limit of large systems. For $\mathcal{V}=\Z$ the strong Szeg\H{o} limit theorem implies the divergence of the bipartite entanglement entropy, albeit without any indication on the behavior as a function of the subsystem's size. For $\mathcal{V}=\N$ an explicit analysis of matrix elements provides a quantitative lower bound depending on the subsystem's size.

\subsection{Properties of $\overline{H}_\Lambda$}
\label{subsec:properties_hamilt}

The orthogonal spectral decomposition of $\overline{h}_\Lambda^{(q)}$ is well known by Fourier analysis and reads
\begin{align}\label{eq:orth_spec_decomp}
	\overline{h}_\Lambda^{(q)}=ODO^T,
\end{align}
with matrix elements
\begin{align*}
	D_{jk} = 4\sin^2\left(\frac{j\pi}{2(|\Lambda|+1)}\right)\delta_{jk},\qquad
	O_{jk} = \sqrt{\frac{2}{|\Lambda|+1}}\sin\left(\frac{jk\pi}{|\Lambda|+1}\right)
\end{align*}
for $j,k\in\lbrace 1,\ldots,|\Lambda|\rbrace$. In particular, $0 < \overline{h}_\Lambda^{(q)} < 4$ for all $\Lambda\in\mathcal{L}$ and thus Assumption~\ref{ass:exhaustive} is satisfied. However, since in the limit $|\Lambda|\to\infty$ the spectrum of $\overline{h}_\Lambda^{(q)}$ covers the whole interval $[0,4]$ and
\begin{equation*}
	\overline{h}_\Lambda := \big(\overline{h}_\Lambda^{(p)}\big)^{1/2} \overline{h}_\Lambda^{(q)}
		\big(\overline{h}_\Lambda^{(p)}\big)^{1/2}
		= \overline{h}_\Lambda^{(q)},
\end{equation*}
inverse powers of $\overline{h}_\Lambda$ are not uniformly bounded in $|\Lambda|$, in particular not on average. Hence, the localization condition \eqref{eq:expectation_decay} fails deterministically. Nevertheless, as the spectrum of $\overline{h}_\Lambda^{(q)}$ does not include $0$ for any $\Lambda\in\mathcal{L}$, the covariance matrix associated to the ground state of $\overline{H}_\Lambda$ is still well defined and given by \eqref{eq:m} as
\begin{align*}
	\overline{\Gamma}_\Lambda = \begin{pmatrix} (\overline{h}_\Lambda^{(q)})^{-1/2} & 0 \\
		0 & (\overline{h}_\Lambda^{(q)})^{1/2}
	\end{pmatrix},
\end{align*}
where $(\overline{h}_\Lambda^{(q)})^{\pm 1/2}$ may be computed using the spectral decomposition \eqref{eq:orth_spec_decomp}.

Following Proposition~\ref{prop:reduction_Gaussian} for any $\Lambda_0\subset\Lambda\in\mathcal{L}$, the reduction of the (Gaussian) ground state of $\overline{H}_\Lambda$ to the tensor component associated with $\Lambda_0$ is still Gaussian and fully characterized in terms of the covariance matrix
\begin{align}\label{eq:cov_matrix_expl_real_red}
	\overline{\Gamma}_{\Lambda_0}
		= \begin{pmatrix} \iota_{\Lambda_0}^{*}(\overline{h}_\Lambda^{(q)})^{-1/2}\iota_{\Lambda_0} & 0 \\
		0 & \iota_{\Lambda_0}^{*}(\overline{h}_\Lambda^{(q)})^{1/2}\iota_{\Lambda_0}
	\end{pmatrix}.
\end{align}
Note that the truncated diagonal blocks in \eqref{eq:cov_matrix_expl_real_red} remain positive definite. With the explicit formula \eqref{eq:bipartite_entanglement} the bipartite entanglement entropy of the ground state is given in terms of the symplectic eigenvalues of $\overline{\Gamma}_{\Lambda_0}$. By Lemma~\ref{corollary:symplectic}, the latter may be computed as the square roots of the eigenvalues of
\begin{align*}
	\left(h_{\Lambda}^{[-1/2]}\right)^{1/2} h_{\Lambda}^{[1/2]} \left(h_{\Lambda}^{[-1/2]}\right)^{1/2}
		\qquad\textnormal{or}\qquad
		\left(h_{\Lambda}^{[1/2]}\right)^{1/2} h_{\Lambda}^{[-1/2]} \left(h_{\Lambda}^{[1/2]}\right)^{1/2},
\end{align*}
where here and in the sequel we use the shorthand 
\begin{align}\label{eq:shorthand_hamiltonian}
	h_{\Lambda}^{[\alpha]} :=
		\iota_{\Lambda_0}^{*}(\overline{h}_\Lambda^{(q)})^\alpha\iota_{\Lambda_0}
		\qquad\textnormal{for } \alpha\in\R.
\end{align}
To close this subsection, we show the following two useful lemmas.
\begin{lemma}\label{lem:eigenvalues_lower_bound}
	Let $k,n\in\N$, $k<n$ and $A\in\R^{k\times k}$, $B\in\R^{n\times n}$ with $A,B>0$. Let furthermore $\iota_k:\R^k\hookrightarrow\R^n$ denote any embedding and set $B^{[\alpha]} := \iota_k^{*}B^\alpha\iota_k$ for any $\alpha\in\R$. Then, we have
	\begin{align}\label{eq:eigenvalue_low_bound}
		\lambda_j \left(B^{[\alpha/2]}A B^{[\alpha/2]}\right)
			\leq \lambda_j\left((B^{[\alpha]})^{1/2}A(B^{[\alpha]})^{1/2}\right)
	\end{align}
	for any $\alpha\in\R$ and all $j=1,\ldots,k$, where $\lambda_1(\cdot)\leq\ldots\leq\lambda_n(\cdot)$ denote eigenvalues in increasing order.
\end{lemma}

\begin{proof}
	Denoting by $P_k = \iota_k\iota_k^{*}$ the orthogonal projection onto $\R^k\subset\R^n$ and exploiting that $P_k\leq\1$, we have
	\begin{align*}
		\left(B^{[\alpha/2]}\right)^2 = \iota_k^{*} B^{\alpha/2} P_k B^{\alpha/2} \iota_k
			\leq \iota_k^{*} B^\alpha \iota_k = B^{[\alpha]}.
	\end{align*}
	Since for any $C\in\C^{n\times n}$ the eigenvalues of $CC^{*}$ and $C^{*}C$ counted with multiplicities coincide, we have for all $j=1,\ldots,k$
	\begin{align*}
		\lambda_j\left(B^{[\alpha/2]}A B^{[\alpha/2]}\right)
			&=\, \lambda_j \left(A^{1/2}\left(B^{[\alpha/2]}\right)^2A^{1/2}\right)
			\leq \lambda_j \left(A^{1/2}B^{[\alpha]}A^{1/2}\right)\\
			&=\, \lambda_j\left((B^{[\alpha]})^{1/2}A(B^{[\alpha]})^{1/2}\right),
	\end{align*}
	which concludes the proof.
\end{proof}
\begin{lemma}\label{lem:lower_bound_det}
	Let $\Lambda\in\mathcal{L}$ and $\overline{H}_\Lambda$ be the Hamiltonian defined in \eqref{eq:hamiltonian_explicit_real}, with ground state density matrix $\overline{\rho}_\Lambda$. For any $\Lambda_0\subset\Lambda$, we have the lower bounds
	\begin{align}
		S(\overline{\rho}_\Lambda;\Lambda_0)
			&\geq\, \frac{1}{2}\max\limits_M\log\left(\det\left(M\right)\right),\label{eq:lower_bound_formula1}\\
		S(\overline{\rho}_\Lambda;\Lambda_0)
			&\geq\, \frac{1}{2}\max\limits_M\log\left(\lambda_\mathrm{max}(M)\right),\label{eq:lower_bound_formula2}
	\end{align}
	where maxima are taken over $M\in\left\lbrace h_{\Lambda}^{[-1/4]}h_{\Lambda}^{[1/2]}h_{\Lambda}^{[-1/4]}, h_{\Lambda}^{[1/4]}h_{\Lambda}^{[-1/2]}h_{\Lambda}^{[1/4]}\right\rbrace$ and $\lambda_\mathrm{max}(M)$ denotes the largest eigenvalue of $M$.
\end{lemma}
\begin{proof}
	The explicit formula \eqref{eq:bipartite_entanglement} and the lower bound from Lemma~\ref{lem:function} \ref{it:lower_bound} imply
	\begin{align}\label{eq:lower:bound_0}
		S(\overline{\rho}_\Lambda;\Lambda_0)
			\geq \sum\limits_{\gamma\in\sigma_\mathrm{symp}(\overline{\Gamma}_{\Lambda_0})} \log(\gamma)
			= \frac{1}{2}\sum\limits_{\gamma\in\sigma(N)}\log(\gamma)
			\geq \max\limits_{\gamma\in\sigma(N)} \frac{1}{2}\log(\gamma)
	\end{align}
	with
	\begin{align*}
		N \in\left\lbrace \left(h_{\Lambda}^{[-1/2]}\right)^{1/2} h_{\Lambda}^{[1/2]} 
			\left(h_{\Lambda}^{[-1/2]}\right)^{1/2},\left(h_{\Lambda}^{[1/2]}\right)^{1/2} h_{\Lambda}^{[-1/2]}
			 \left(h_{\Lambda}^{[1/2]}\right)^{1/2}\right\rbrace.
	\end{align*}
	Here, $\sigma(N)$ stands for the spectrum of $N$ and all sums over $\sigma(N)$ are meant with multiplicities. The equality in \eqref{eq:lower:bound_0} uses the characterization of the symplectic spectrum of $\overline{\Gamma}_{\Lambda_0}$ given in Lemma~\ref{corollary:symplectic}; while the last inequality is a consequence of the Heisenberg uncertainty relations in Proposition~\ref{prop:williamson}, implying $\sigma(N)\subset [1,\infty)$.	We complete the proof by relating the eigenvalues of $M$ and $N$ via Lemma~\ref{lem:eigenvalues_lower_bound} and using $\sum\limits_{\gamma\in\sigma(M)}\log(\gamma) = \log(\det(M))$.
\end{proof}

\subsection{Proof of Theorem~\ref{theorem:entropy_grow} for $\mathcal{V}=\Z$}
\label{subsec:proof_Z}

Either for $\mathcal{V}=\Z$ or for $\mathcal{V}=\N$, a first step towards a proof of Theorem~\ref{theorem:entropy_grow} consists in noticing that for suitable $\Lambda_0\subset\Lambda$ and large enough $\Lambda$ some of the matrices $h_{\Lambda}^{[\alpha]}$ take, up to small corrections, a simple form.
\begin{lemma}\label{lem:hamilton_estimate_Z}
	For any finite and connected $\Lambda_0\subset\Lambda\subset\Z$ let $h_{\Lambda}^{[\alpha]}$ be the matrix~\eqref{eq:shorthand_hamiltonian}. Let moreover for any $\alpha > -1/2$ the sequence $h_\Z^{[\alpha]}$ be defined by its entries
	\begin{align}\label{eq:hamiltonian_Z}
		\left(h_\Z^{[\alpha]}\right)_{jk}
			= 4^\alpha\int_0^1\sin^{2\alpha}\left(\frac{\pi x}{2}\right)\cos(|j-k|\pi x)\,\mathrm{d}x,
			\qquad j,k\in\Z.
	\end{align}
	For $\Lambda_0 = [-n,n]\cap\Z$ and $\Lambda = [-m,m]\cap\Z$ with $n,m\in\N$, $n\leq m$, we then have
	\begin{align}
		\left(h_{\Lambda}^{[\alpha]}\right)_{jk} &=\, \left(h_\Z^{[\alpha]}\right)_{jk}
			+ \mathcal{O}\left(\frac{|\Lambda_0|}{|\Lambda|}\right),&&\alpha \geq 0\label{eq:hamilton_Z_1/2}\\
		\left(h_{\Lambda}^{[\alpha]}\right)_{jk} &=\, \left(h_\Z^{[\alpha]}\right)_{jk}
			+ \mathcal{O}\left(\frac{|\Lambda|^{-2\alpha}+|\Lambda_0|^{1-2\alpha}}{|\Lambda|}\right),
			&&\alpha\in(-1/2,0)\label{eq:hamilton_Z_-1/4}
	\end{align}
	with error terms uniform in $j,k\in\lbrace 1,\ldots,|\Lambda_0|\rbrace$.
\end{lemma}
\begin{proof}
	By \eqref{eq:shorthand_hamiltonian} the matrix $h_{\Lambda}^{[\alpha]}$ is the truncation of the matrix $(h_\Lambda^{(q)})^\alpha$ according to the inclusion $ \Lambda_0 \subset \Lambda $. We index the matrix elements of $h_{\Lambda}^{[\alpha]}$ and $(h_\Lambda^{(q)})^\alpha$ by $j,k\in\lbrace 1,\ldots,|\Lambda_0|\rbrace$ and $\bar{j},\bar{k}\in\lbrace 1,\ldots,|\Lambda|\rbrace$, respectively. In particular, for the special choices $\Lambda_0 = [-n,n]\cap\Z$ and $\Lambda = [-m,m]\cap\Z$ with $n,m\in\N$, $n\leq m$, we have the relation
	\begin{equation*}
		\left(h_{\Lambda}^{[\alpha]}\right)_{jk} = \left((h_\Lambda^{(q)})^\alpha\right)_{\bar{j}\bar{k}}
			\quad\textnormal{with}\quad
			\lbrace\bar{j},\bar{k}\rbrace = \lbrace j,k\rbrace + \frac{1}{2}(|\Lambda|-|\Lambda_0|)
	\end{equation*}
	for all $j,k\in\lbrace 1,\ldots,|\Lambda_0|\rbrace$. By spectral calculus and the explicit diagonalization~\eqref{eq:orth_spec_decomp}, the matrix elements of $h_{\Lambda}^{[\alpha]}$ may now be written as
	\begin{align}
		\left(h_{\Lambda}^{[\alpha]}\right)_{jk}
			=&\,\frac{4^{\alpha}}{|\Lambda|+1} \sum\limits_{l=1}^{|\Lambda|}
				\sin^{2\alpha}\left(\frac{l\pi}{2(|\Lambda|+1)}\right)\cos\left(\frac{|j-k|l\pi}{|\Lambda|+1}\right)
				\label{eq:hamiltonian_Z_explicit_1}\\
			&\,-\frac{4^{\alpha}}{|\Lambda|+1} \sum\limits_{l=1}^{|\Lambda|}
				\sin^{2\alpha}\left(\frac{l\pi}{2(|\Lambda|+1)}\right)
				\cos\left(\frac{(j+k+|\Lambda|-|\Lambda_0|)l\pi}{|\Lambda|+1}\right),\label{eq:hamiltonian_Z_explicit_2}
	\end{align}
	where we used $2\sin(a)\sin(b)=\cos(a-b)-\cos(a+b)$. For $\alpha > -1/2$ we show that (i) the Riemann sum in \eqref{eq:hamiltonian_Z_explicit_1} converges to the integral \eqref{eq:hamiltonian_Z} while (ii) the sum \eqref{eq:hamiltonian_Z_explicit_2} vanishes, both with the rates given in \eqref{eq:hamilton_Z_1/2} and \eqref{eq:hamilton_Z_-1/4}.
	\medskip\\
	\noindent(i) The function $(0,1]\ni x\mapsto f(x):=\sin^{2\alpha}\left(\tfrac{\pi x}{2}\right)\cos(|j-k|\pi x)$ is of bounded variation for all $\alpha\geq 0$ and all $j,k\in\lbrace 1,\ldots,|\Lambda_0|\rbrace$. Integrating the modulus of the derivative, its total variation is at most of order $\mathcal{O}(|\Lambda_0|)$, whence the convergence to the integral \eqref{eq:hamiltonian_Z} occurs with rate $\mathcal{O}(|\Lambda_0|/|\Lambda|)$.
	 
	 For $\alpha\in(-1/2,0)$ the function $f$ is positive and strictly monotone decreasing either on all $(0,1]$ if $j=k$, or up to its first zero at $x=(2|j-k|)^{-1}$ if $j\neq k$. Setting $\beta=1$ for $j=k$ and $\beta= (2|j-k|)^{-1}$ for $j\neq k$, the rate of convergence of the Riemann sum on $(0,\beta]$ is at most of order $\mathcal{O}(|\Lambda|^{-1-2\alpha})$. In fact, setting $n=|\Lambda|+1$ and using the bound $f(x)\leq x^{2\alpha}$ for any $x\in(0,\beta]$, we obtain
	 \begin{align*}
	 	&\left\vert \frac{1}{n}\sum\limits_{k=1}^{\lfloor n\beta\rfloor}f\left(\frac{k}{n}\right)
	 		-\int_0^{\beta}f(x)\,\mathrm{d}x\right\vert
	 		= \sum\limits_{k=1}^{\lfloor n\beta\rfloor}
	 			\int_{\tfrac{k-1}{n}}^{\tfrac{k}{n}} \left(f(x)-f\left(\frac{k}{n}\right)\right)\mathrm{d}x
	 			+ \int_{\lfloor n\beta\rfloor/n}^\beta \mkern-10mu f(x) dx\\
	 		&\leq\, \int_0^{\tfrac{1}{n}}\left(f(x)-f\left(\frac{1}{n}\right)\right)\mathrm{d}x
	 			+ \frac{1}{n}\sum\limits_{k=2}^{\lfloor n\beta\rfloor}\left(
	 			f\left(\frac{k-1}{n}\right)-f\left(\frac{k}{n}\right)\right)
	 			+\frac{1}{n} f\left(\frac{\lfloor n\beta\rfloor}{n}\right) \\
	 		&\leq \int_0^{\tfrac{1}{n}} x^{2\alpha}\mathrm{d}x
	 		= \mathcal{O}\left(|\Lambda|^{-1-2\alpha}\right).
	 \end{align*} 
	 Here, $\lfloor n\beta\rfloor$ denotes the integer part of $n\beta$.
	 
	 On $(\beta,1]$ the function $f$ is Lipschitz continuous. Bounding its derivative, the associated constant is of order $\mathcal{O}(|\Lambda_0|^{1-2\alpha})$, whence the convergence rate of the Riemann sum on this interval is at most of order $\mathcal{O}(|\Lambda_0|^{1-2\alpha}/|\Lambda|)$.
	 \medskip\\
	 \noindent(ii) Defining the function $(0,1]\ni x\mapsto g(x) = \sin^{2\alpha}\left(\tfrac{\pi x}{2}\right)\cos((j+k-1-|\Lambda_0|)\pi x)$ and using the identity $\cos(a+l\pi) = (-1)^l\cos(a)$ for any $l\in\Z$, the Riemann sum in \eqref{eq:hamiltonian_Z_explicit_2} reads
	 \begin{align}
	 	\frac{4^{\alpha}}{|\Lambda|+1} &\sum\limits_{l=1}^{|\Lambda|}(-1)^l g\left(\frac{l}{|\Lambda|+1}\right)\nonumber\\
	 		&=\, \frac{4^\alpha}{|\Lambda|+1}\sum\limits_{m=1}^{\tfrac{|\Lambda|-1}{2}}
	 			g\left(\frac{2m}{|\Lambda|+1}\right)
	 			- \frac{4^\alpha}{|\Lambda|+1}\sum\limits_{m=0}^{\tfrac{|\Lambda|-1}{2}}
	 			g\left(\frac{2m+1}{|\Lambda|+1}\right).\label{eq:double_riemann_sum}
	 \end{align}
	 Each of the sums in \eqref{eq:double_riemann_sum} converges as $|\Lambda|\to\infty$ to the integral $2^{2\alpha - 1}\int_0^1 g(x)\mathrm{d}x$. From item (i), the convergence rates are those given in \eqref{eq:hamilton_Z_1/2} and \eqref{eq:hamilton_Z_-1/4}, depending on whether $\alpha\geq 0$ or $\alpha\in(-1/2,0)$. 
\end{proof}
Since $(h_\Z^{[\alpha]})_{jk}$ only depends on $|j-k|$, any restriction of the index set $j,k\in\Z$ to some finite, connected subset $\Lambda\subset\Z$ defines a symmetric \emph{Toeplitz matrix}. Determinants of increasingly large Toeplitz matrices satisfy well-known relations, as we briefly detail. Let $\widehat{\phi}:\Z\to\R$ be an even function, so that for any $n\in\N$
\begin{align*}
	T_n(\phi) = \left\lbrace \widehat{\phi}(j-k)\right\rbrace_{j,k=1}^n
\end{align*}
is a symmetric $n\times n$ Toeplitz matrix. The sequence $\lbrace\widehat{\phi}(k)\rbrace_{k\in\Z}$ gives the Fourier coefficients of a function $\phi:[-\pi,\pi]\to\R$, referred to as the \emph{symbol} of the Toeplitz matrices $\lbrace T_n(\phi)\rbrace_{n\in\N}$, where we use the prefactor convention:
\begin{align}\label{eq:Toeplitz_symbol}
	\widehat{\phi}(k) = \frac{1}{2\pi}\int_{-\pi}^{\pi} \phi(x) e^{-ikx}\,\mathrm{d}x,\qquad
		\phi(x) = \sum_{k\in\Z}\widehat{\phi}(k) e^{ikx}.
\end{align}
The following result on the asymptotic behavior of $\det(T_n(\phi))$ as $n\to\infty$ is known as \emph{Szeg\H{o}'s second limit theorem} or \emph{strong Szeg\H{o} limit theorem}, which we state in a version due to Ibragimov \cite{Ibragimov1968}.
\begin{proposition}[Szeg\H{o}'s second limit theorem]\label{theorem:szego}
	Let $\lbrace T_n(\phi)\rbrace_{n\in\N}$ be a sequence of Toeplitz matrices with symbol $\phi$ and set
	\begin{align*}
		G(\phi) = \exp\left((\widehat{\log\circ\phi})(0)\right).
	\end{align*}
	If $\phi\in L^1([-\pi,\pi])$ and $\log\circ\phi\in L^1([-\pi,\pi])$, we have
	\begin{align*}
		\lim\limits_{n\to\infty} \frac{\det\left(T_n(\phi)\right)}{G(\phi)^n}
			= \exp\left(\sum\limits_{k=1}^\infty k \left\vert(\widehat{\log\circ\phi})(k)\right\vert^2\right).
	\end{align*}
	In particular, if the right-hand side diverges, then so does the left-hand side.
\end{proposition}
\begin{proof}[Proof of Theorem~\ref{theorem:entropy_grow} for $\mathcal{V}=\Z$]
	By Lemma~\ref{lem:lower_bound_det}, the entanglement entropy of the ground state of $\overline{H}_\Lambda$ with respect to the bipartition $\Lambda = \Lambda_0\cup\left(\Lambda\setminus\Lambda_0\right)$ has the lower bound
	\begin{align*}
		S\left(\overline{\rho}_{\Lambda};\Lambda_0\right)
			\geq \frac{1}{2}\log\det\left(h_{\Lambda}^{[1/2]}\right)
				+ \log\det\left(h_{\Lambda}^{[-1/4]}\right).
	\end{align*}
	We show that for suitable choices of $\Lambda_0\subset\Lambda$ both terms on the right-hand side diverge as $|\Lambda|, |\Lambda_0|\to\infty$. To make use of the estimates from Lemma~\ref{lem:hamilton_estimate_Z} the following perturbative argument applies. Let $A,B>0$ be two $|\Lambda_0|\times|\Lambda_0|$ matrices, then
	\begin{align*}
		\log\det(A+B) &=\, \log\det(A) + \trace\log(\1+A^{-1/2}BA^{-1/2})\\
			&\leq\, \log\det(A) +|\Lambda_0|\Vert A^{-1}\Vert \Vert  B\Vert,
	\end{align*}
	where we used $\log\det(C)=\trace\log(C)$ for any (finite) matrix $C>0$ and $\log(1+x)\leq x$ for all $x\in(-1,\infty)$.
	
	Let $h_\Z^{[\alpha]}$ stand in this proof for the $|\Lambda_0|\times|\Lambda_0|$ matrix from \eqref{eq:hamilton_Z_1/2} or \eqref{eq:hamilton_Z_-1/4}. Setting $B_1 = h_\Z^{[1/2]} - h_{\Lambda}^{[1/2]}$ and $B_2 = h_\Z^{[-1/4]} - h_{\Lambda}^{[-1/4]}$ we obtain
	\begin{align}
		S\left(\overline{\rho}_{\Lambda};\Lambda_0\right)
			\geq &\,\frac{1}{2}\log\det\left(h_\Z^{[1/2]}\right) + \log\det\left(h_\Z^{[-1/4]}\right)
				\nonumber\\
				&-\,|\Lambda_0|\left(\Vert (h_{\Lambda}^{[1/2]})^{-1}\Vert \Vert B_1\Vert
				+\Vert (h_{\Lambda}^{[-1/4]})^{-1}\Vert \Vert B_2\Vert\right).\label{eq:entropy_lower_bound_error}
	\end{align}
	For any $\alpha > -1/2$ a comparison between \eqref{eq:hamiltonian_Z} and \eqref{eq:Toeplitz_symbol} shows that $\lbrace h_\Z^{[\alpha]}\rbrace_{|\Lambda_0|\in\N}$ defines a sequence of Toeplitz matrices with symbol $\phi_\alpha(x):=|2\sin(x/2)|^{2\alpha}$. An explicit computation yields
	\begin{align*}
		(\widehat{\log\circ\phi_\alpha})(k)
			= 2\alpha (\widehat{\log\circ\phi_{1/2}})(k)
			=  \begin{cases} 0, & \textnormal{for }k = 0,\\ -\frac{\alpha}{k}, & \textnormal{for }k\in\N. \end{cases}
	\end{align*}
	Applying Proposition~\ref{theorem:szego} with $G(\phi_\alpha) = 1$ and, since for any $\alpha\in(-1/2,\infty)\setminus\lbrace 0\rbrace$
	\begin{align*}
		\sum\limits_{k=1}^{N} k\left\vert(\widehat{\log\circ\phi_\alpha})(k)\right\vert^2
			= \sum\limits_{k=1}^{N} \frac{|\alpha|^2}{k}\quad \xrightarrow{N\to\infty}\quad\infty,
	\end{align*}
	$\det(h_\Z^{[1/2]})$ and $\det(h_\Z^{[-1/4]})$ diverge as $n\to\infty$ for any sequence of finite and connected subsets $\lbrace\Lambda_0^{(n)}\rbrace_{n\in\N}$ of $\Z$ with $|\Lambda_0^{(n)}|\to\infty$.
	
	It remains to handle the error term~\eqref{eq:entropy_lower_bound_error}. For any $n, m_n\in\N$ let us introduce the notation $\Lambda_0^{(n)}=[-n,n]\cap\Z$ and $\Lambda^{(n)}= [-m_n,m_n]\cap\Z$. In the sequel, we consider sequences of subsets $\lbrace\Lambda_0^{(n)}\rbrace_{n\in\N},\lbrace\Lambda^{(n)}\rbrace_{n\in\N}$ with $n\leq m_n$ to be specified later and the associated entanglement entropy $S(\overline{\rho}_{\Lambda^{(n)}};\Lambda_0^{(n)})$. By Lemma~\ref{lem:hamilton_estimate_Z}, $B_1$ and $B_2$ in \eqref{eq:entropy_lower_bound_error} satisfy
	\begin{align*}
		\Vert B_1\Vert =\mathcal{O}\left(\frac{|\Lambda_0^{(n)}|^2}{|\Lambda^{(n)}|}\right)
			\qquad\textnormal{and}\qquad
			\Vert B_2\Vert = \mathcal{O}\left(\frac{|\Lambda_0^{(n)}|}{|\Lambda^{(n)}|^{1/2}}
			+ \frac{|\Lambda_0^{(n)}|^{5/2}}{|\Lambda^{(n)}|}\right).
	\end{align*}
	Since the spectrum of $\overline{h}_\Lambda^{(q)}$ is contained in $(0,4)$ for any finite, connected $\Lambda\subset\Z$, we deduce from \eqref{eq:shorthand_hamiltonian} that $\Vert(h_{\Lambda}^{[-1/4]})^{-1}\Vert \leq 1/\sqrt{2}$ for any $\Lambda_0\subset\Lambda$. As for $\Vert(h_{\Lambda^{(n)}}^{[1/2]})^{-1}\Vert$, we first observe that by \eqref{eq:hamiltonian_Z}
	\begin{align*}
		\left(h_\Z^{[1/2]}\right)_{jk} = -\frac{1}{\pi}\frac{1}{(j-k)^2 - 1/4},\qquad
			\textnormal{for all }j,k\in\lbrace 1,\ldots,|\Lambda_0^{(n)}|\rbrace.
	\end{align*}
	By Gershgorin's circle theorem (see, e.g.~\cite{Bhatia1996}), $h_\Z^{[1/2]}$ has thus the lower bound
	\begin{align*}
		\min\limits_{1\leq j\leq |\Lambda_0^{(n)}|} \left(\left(h_\Z^{[1/2]}\right)_{jj}
			- \sum\limits_{k\neq j} \left\vert \left(h_\Z^{[1/2]}\right)_{jk}\right\vert\right)
			\geq \frac{4}{\pi} - \frac{2}{\pi} \sum\limits_{l=1}^{|\Lambda_0^{(n)}|}\frac{1}{l^2-1/4}
			= \mathcal{O}\left(|\Lambda_0^{(n)}|^{-1}\right),
	\end{align*}
	whence also $h_{\Lambda^{(n)}}^{[1/2]}\geq\mathcal{O}(|\Lambda_0^{(n)}|^{-1})$, since the correction of order $\mathcal{O}(|\Lambda_0^{(n)}|^2/|\Lambda^{(n)}|)$ can be chosen arbitrarily small. We conclude that $\Vert(h_{\Lambda^{(n)}}^{[1/2]})^{-1}\Vert = \mathcal{O}(|\Lambda_0^{(n)}|)$ and thus the error term~\eqref{eq:entropy_lower_bound_error} is at most of order $\mathcal{O}(|\Lambda_0^{(n)}|/|\Lambda^{(n)}|^{1/2})+\mathcal{O}(|\Lambda_0^{(n)}|^3/|\Lambda^{(n)}|)$. This vanishes in the limit $n\to\infty$ whenever $m_n$ is chosen such that $m_n =\mathcal{O}( n^{3+\epsilon})$ for some $\epsilon>0$.
\end{proof}

\subsection{Proof of Theorem~\ref{theorem:entropy_grow} for $\mathcal{V}=\N$}
\label{subsec:proof_N}

Let us begin with the counterpart of Lemma~\ref{lem:hamilton_estimate_Z} for $\mathcal{V}=\N$.
\begin{lemma}\label{lem:hamilton_estimate_N}
	For any finite and connected $\Lambda_0\subset\Lambda\subset\N$ let $h_{\Lambda}^{[\alpha]}$ be the matrix~\eqref{eq:shorthand_hamiltonian}. Let moreover for $\alpha\geq -1/2$ the sequence $h_\N^{[\alpha]}$ be given by its entries
	\begin{align}\label{eq:hamiltonian_N}
		\left(h_\N^{[\alpha]}\right)_{jk}
			= 2^{1+2\alpha}\int_0^1\sin\left(\frac{\pi x}{2}\right)^{2\alpha}\sin(j\pi x)\sin(k\pi x)\,\mathrm{d}x,
			\qquad j,k\in\N.
	\end{align}
	For $\Lambda_0 = [1,n]\cap\N$ and $\Lambda = [1,m]\cap\N$ with $n,m\in\N$, $n\leq m$, we then have
	\begin{align}
		\left(h_{\Lambda}^{[\alpha]}\right)_{jk} &=\, \left(h_\N^{[\alpha]}\right)_{jk}
			+ \mathcal{O}\left(\frac{|\Lambda_0|}{|\Lambda|}\right),&&\alpha \geq 0\label{eq:hamilton_N_approx_1}\\
		\left(h_{\Lambda}^{[\alpha]}\right)_{jk} &=\, \left(h_\N^{[\alpha]}\right)_{jk}
			+ \mathcal{O}\left(\frac{|\Lambda_0|}{|\Lambda|}^{1-2\alpha}\right),&&\alpha\in[-1/2,0)\label{eq:hamilton_N_approx_2}	
	\end{align}
	with error terms uniform in $j,k\in\lbrace 1,\ldots,|\Lambda_0|\rbrace$.
\end{lemma}
\begin{remark}
The integral \eqref{eq:hamiltonian_N} is in fact well defined for $\alpha>-3/2$ and the integrand Lipschitz continuous for $\alpha\geq -1$. \end{remark}
\begin{proof}
	Using spectral calculus and the explicit diagonalization~\eqref{eq:orth_spec_decomp} the matrix elements of $h_{\Lambda}^{[\alpha]}$ read
	\begin{align*}
		\left(h_{\Lambda}^{[\alpha]}\right)_{jk}
			= \frac{2^{1+2\alpha}}{|\Lambda|+1} \sum\limits_{l=1}^{|\Lambda|}
				\sin\left(\frac{l\pi}{2(|\Lambda|+1)}\right)^{2\alpha}\sin\left(\frac{jl\pi}{|\Lambda|+1}\right)
				\sin\left(\frac{kl\pi}{|\Lambda|+1}\right)
	\end{align*}
	for all $j,k\in\lbrace 1,\ldots,|\Lambda_0|\rbrace$. This converges as a Riemann sum to the integral~\eqref{eq:hamiltonian_N} for any $\alpha\geq -1/2$. As the function
	\begin{align*}
		[0,1]\ni x\mapsto \sin\left(\frac{\pi x}{2}\right)^{2\alpha}\sin(j\pi x)\sin(k\pi x)
	\end{align*}
	is Lipschitz continuous with constant denoted by $C_{\alpha}(j,k)$, the rate of convergence of the Riemann sum is at least of order $\mathcal{O}(C_{\alpha}(j,k)/|\Lambda|)$. By elementary bounds on the derivatives there exist $C_1,C_2\in(0,\infty)$ such that $C_\alpha(j,k)\leq C_1|\Lambda_0|$ for $\alpha\geq 0$ and $C_\alpha(j,k)\leq C_2|\Lambda_0|^{1-2\alpha}$ for $\alpha\in[-1/2,0)$ and all $j,k\in\lbrace 1,\ldots,|\Lambda_0|\rbrace$.
\end{proof}
Our analysis requires a further understanding of $h_\N^{[\alpha]}$. To ease notation we use the shorthands
\begin{align}\label{eq:shorthand_RS}
	R_{jk}=\left(h_\N^{[1/4]}\right)_{jk}\qquad\textnormal{and}\qquad S_{jk}=\left(h_\N^{[-1/2]}\right)_{jk}
\end{align}
for all $j,k\in\lbrace 1,\ldots,|\Lambda_0|\rbrace$. Explicit expressions for $R_{jk}$ and $S_{jk}$ are given in Lemma~\ref{lem:matrix_elements_N}. Here, we content ourselves with the following list of properties needed for the proof of Theorem~\ref{theorem:entropy_grow} below.
\begin{lemma}\label{lem:properties_RS}
	Let $n\in\N\setminus\{1\}$ and $j,k\in\lbrace 1,\ldots,n\rbrace$. Then, we have
	\begin{enumerate}[label=(\roman*)]
	\item\label{it:RS_1} $R_{jj}>0$, $R_{jk}<0$ if $j\neq k$ and $|R_{jk}|\leq 2\sqrt{2}$.
	\item\label{it:RS_1bis} $R_{nn-1}$ is a decreasing function of $n\in\N\setminus\{1\}$.
	\item\label{it:RS_2} $2\sum\limits_{k=1}^{n-1}R_{kn}\geq -R_{nn}$.
	\item\label{it:RS_3} $S_{jk} > \frac{1}{\pi}\log\left(\frac{j+k+\tfrac{1}{2}}{|j-k|+\tfrac{1}{2}}\right)>0$.
	\item\label{it:RS_4} $S_{jk} < \frac{1}{\pi}\log\left(\frac{j+k-\tfrac{1}{2}}{|j-k|-\tfrac{1}{2}}\right)$ for $j\neq k$ and $S_{jj}< \frac{2}{\pi} + \frac{1}{\pi}\log(4j-1)$.
	\item\label{it:RS_5} $S_{jn}$ is an increasing function of $1\leq j\leq n$.
	\end{enumerate}
\end{lemma}
\begin{proof}
	\noindent \ref{it:RS_1}-\ref{it:RS_1bis} $R_{jj} >0$ and $|R_{jk}|\leq 2\sqrt{2}$ follow directly from the definition \eqref{eq:hamiltonian_N} with $\alpha=1/4$. The remaining property, $R_{jk}<0$ if $j\neq k$, is immediate if one rewrites the explicit formula from Lemma~\ref{lem:matrix_elements_N}\ref{it:matel_1} in the form
	\begin{align*}
		\frac{1}{2\sqrt{2\pi}}\left[\left(\prod\limits_{l=1}^{2\min\lbrace j,k\rbrace}
			\frac{(j+k-l-\tfrac{1}{4})}{(j+k-l+\tfrac{5}{4})}\right)-1\right]
			\frac{\Gamma(|j-k|-\tfrac{1}{4})}{\Gamma(|j-k|+\tfrac{5}{4})},
	\end{align*}
	where we used the recursion $\Gamma(n+1)=n\Gamma(n)$. In fact, the quotient of Gamma functions is positive for $j\neq k$ and the square bracket is negative. Setting $j=n$ and $k=n-1$, item \ref{it:RS_1bis} follows.

	\noindent \ref{it:RS_2} For any $q>0$ we have the identity \cite{Mathematica2018}
	\begin{align*}
		\sum\limits_{k=1}^\infty \frac{\Gamma(k-\tfrac{q}{2})}{\Gamma(k+1+\tfrac{q}{2})}
			= \frac{\Gamma(1-\tfrac{q}{2})}{q\,\Gamma(1+\tfrac{q}{2})}.
	\end{align*}
	Hence, estimating the formula from Lemma~\ref{lem:matrix_elements_N}\ref{it:matel_1} by discarding positive or negative terms, we obtain on the one hand
	\begin{align*}
		2\sqrt{2\pi}\sum\limits_{k=1}^{n-1} R_{kn} \geq 
			-\sum\limits_{k=1}^{n-1} \frac{\Gamma(n-k-\tfrac{1}{4})}{\Gamma(n-k+\tfrac{5}{4})}
			\geq -\sum\limits_{k=1}^\infty \frac{\Gamma(k-\tfrac{1}{4})}{\Gamma(k+\tfrac{5}{4})}
			= -2\frac{\Gamma(\tfrac{3}{4})}{\Gamma(\tfrac{5}{4})};
	\end{align*}
	and on the other hand
	\begin{align*}
		-\sqrt{2\pi} R_{nn}
			= -\frac{1}{2}\left(\frac{\Gamma(2n-\tfrac{1}{4})}{\Gamma(2n+\tfrac{5}{4})}
			-\frac{\Gamma(-\tfrac{1}{4})}{\Gamma(\tfrac{5}{4})}\right)
			\leq \frac{1}{2}\frac{\Gamma(-\frac{1}{4})}{\Gamma(\tfrac{5}{4})}
			= -2\frac{\Gamma(\tfrac{3}{4})}{\Gamma(\tfrac{5}{4})}.
	\end{align*}
	
	\noindent \ref{it:RS_3}-\ref{it:RS_5} Using Lemma~\ref{lem:matrix_elements_N}\ref{it:matel_2} we have
	\begin{align*}
		\pi S_{jk} = \sum\limits_{l=|j-k|}^{j+k-1}\frac{1}{l+\tfrac{1}{2}}
			> \int_{|j-k|}^{j+k}\frac{\mathrm{d}x}{x+\tfrac{1}{2}}
			= \log\left(\frac{j+k+\tfrac{1}{2}}{|j-k|+\tfrac{1}{2}}\right)>0.
	\end{align*}
	The upper bounds on $S_{jk}$ follow in the same manner, though with a distinction for $j\neq k$ and $j= k$. Item \ref{it:RS_5} is immediate from the explicit formula for $S_{jn}$.
\end{proof}
\begin{proof}[Proof of Theorem~\ref{theorem:entropy_grow} for $\mathcal{V}=\N$]
	By Lemma~\ref{lem:lower_bound_det}, the entanglement entropy of the ground state of $\overline{H}_\Lambda$ with respect to the bipartition $\Lambda = \Lambda_0\cup\left(\Lambda\setminus\Lambda_0\right)$ has the lower bound
	\begin{align*}
		S\left(\overline{\rho}_{\Lambda};\Lambda_0\right)
			\geq \frac{1}{2}\log\left[ \lambda_\mathrm{max}
			\left(h_{\Lambda}^{[1/4]} h_{\Lambda}^{[-1/2]} h_{\Lambda}^{[1/4]}\right)\right],
	\end{align*}
	where $\lambda_\mathrm{max}(A)$ denotes the largest eigenvalue of the matrix $A$. By the min-max principle, this implies in particular
	\begin{align}\label{eq:entropy_lower_bound_matel}
		S\left(\overline{\rho}_{\Lambda};\Lambda_0\right)
			\geq \frac{1}{2}\log\left[
				\left(h_{\Lambda}^{[1/4]} h_{\Lambda}^{[-1/2]} h_{\Lambda}^{[1/4]}\right)_{|\Lambda_0||\Lambda_0|}\right].
	\end{align}
	Let $\Lambda_0 = [1,n]\cap\N$ and $\Lambda = [1,m]\cap\N$ with $n,m\in\N$ and $n\leq m$. By Lemma~\ref{lem:hamilton_estimate_N} the $nn$-matrix element on the right hand side of~\eqref{eq:entropy_lower_bound_matel} satisfies
	\begin{align}\label{eq:main_and_error_terms}
		\left(h_{\Lambda}^{[1/4]} h_{\Lambda}^{[-1/2]} h_{\Lambda}^{[1/4]}\right)_{nn}
			= \sum\limits_{j,k=1}^{n} R_{nj}S_{jk}R_{kn} + \mathcal{O}\left(\frac{n^4}{m}\right),
	\end{align}
	where we used the shorthand notation \eqref{eq:shorthand_RS}. In fact, by Lemma~\ref{lem:properties_RS} all $R_{jk}$ are bounded and $|S_{jk}|\leq\mathcal{O}(\log n)$. The leading contribution to the error thus arises from the correction of order $\mathcal{O}(n^2/m)$ to $S$ and two additional powers of $n$ from the sum over $j$ and $k$.
	
	As $S$ is positive definite, so is $\sum_{j,k=1}^l  R_{nj}S_{jk}R_{kn} \geq 0$ for any $1\leq l\leq n$. Choosing $l=n-2$ this yields
	\begin{align}
		(RSR)_{nn} \geq&\, \left(R_{nn}\right)^2 S_{nn} + \left(R_{n n-1}\right)^2 S_{n-1 n-1}\nonumber\\
			&+\, 2 R_{nn}\sum\limits_{j=1}^{n-1} R_{n j} S_{jn}
			+ 2 R_{n n-1}\sum\limits_{j=1}^{n-2} R_{n j} S_{j n-1}.\label{eq:RSR_lower_bound}
	\end{align}
	We proceed by repeated use of Lemma~\ref{lem:properties_RS}. By \ref{it:RS_1} and \ref{it:RS_3} the last term in \eqref{eq:RSR_lower_bound} is positive and may thus be discarded for a lower bound. Similarly, the second but last term is negative and satisfies by \ref{it:RS_2} and \ref{it:RS_5}
	\begin{align*}
		2 R_{nn}\sum\limits_{j=1}^{n-1} R_{n j} S_{jn}
			\geq 2 R_{nn}S_{n n-1}\sum\limits_{j=1}^{n-1} R_{jn}
			\geq - \left(R_{nn}\right)^2 S_{n n-1}.
	\end{align*}
	Finally, we obtain
	\begin{align}
		(RSR)_{nn} &\geq\, \left(R_{nn}\right)^2 \left[S_{nn} - S_{n n-1}\right] + \left(R_{n n-1}\right)^2 S_{n-1 n-1}
			\nonumber\\
			&\geq\, \frac{1}{\pi}\left(R_{2 1}\right)^2 \log(4n-3),\label{eq:lower_bound_final}
	\end{align}
	where the square bracket is positive by \ref{it:RS_5} and the last inequality follows by \ref{it:RS_1bis} and \ref{it:RS_3}.
	
	Consider now the sequences $\lbrace\Lambda_0^{(n)}\rbrace_{n\in\N}$ and $\lbrace\Lambda^{(n)}\rbrace_{n\in\N}$ with $\Lambda_0^{(n)} = [1,n]\cap\N$ and $\Lambda^{(n)} = [1,m_n]\cap\N$, where $m_n$ is chosen such that $m_n =\mathcal{O}(n^{4+\epsilon})$ for some $\epsilon>0$. Then, combining \eqref{eq:entropy_lower_bound_matel}, \eqref{eq:main_and_error_terms} and \eqref{eq:lower_bound_final}, the entanglement entropy $S(\overline{\rho}_{\Lambda^{(n)}};\Lambda_0^{(n)})$ diverges as $\log\log(n)$ for $n\to\infty$.
\end{proof}

\appendix
\section{On Gaussian states and entanglement entropy}
\label{app:gaussian_entropy}

This appendix aims at giving a short overview of the fundamentals on Gaussian states and their entanglement entropy. The interested reader is refered to more comprehensive works on the subject for further details \cite{Bratteli1981, Manuceau1968}.

\subsection{Gaussian states}\label{ssec:gaussian}

Let $\mathcal{H}$ be a (separable) Hilbert space. A state $\omega$ on $\mathcal{H}$ is a positive linear functional of norm $1$ on $\mathcal{B}(\mathcal{H})$, the Banach space of bounded operators on $\mathcal{H}$. In particular, to each density matrix $\rho$---i.e. positive semidefinite hermitian operator of trace $1$---corresponds a state $\omega_\rho$ through the identification
\begin{equation*}
	\omega_\rho(A) := \Tr{\rho A}\qquad\textnormal{for all }A\in\mathcal{B}(\mathcal{H}).
\end{equation*}
It suffices to characterize a state by its action on a dense set of bounded operators. One such is given by a  suitable representation of the \emph{Weyl algebra} $\mathcal{W}(V)$, for some linear space $V$ endowed with a real symplectic form $\sigma$. The Weyl algebra is (uniquely) characterized in terms of its elements and generators $W(f), f\in V$ by
\begin{equation*}
	W(f)^*=W(-f)\quad\text{and}\quad W(f)W(g)=e^{-i\sigma(f,g)/2}W(f+g), \qquad\textnormal{for all }f,g\in V.
\end{equation*}
A functional $\omega: \mathcal{W}(V)\to\C$ is called \emph{quasi-free} or \emph{Gaussian} if up to automorphism of the Weyl algebra $\mathcal{W}(V)$ we have
\begin{equation*}
	\omega(W(f)) = e^{-s(f,f)/4}\qquad\textnormal{for all }f\in V,
\end{equation*}
where $s$ is a real symmetric bilinear positive semidefinite form. A quasi-free functional is a state if and only if for the underlying symplectic form
\begin{equation*}
	\sigma(f,g)^2 \leq s(f,f)s(g,g)\qquad\textnormal{for all }f,g\in V.
\end{equation*}
We turn to the concrete setting described in Section~\ref{sec:model}. A realization of the Weyl algebra on the Hilbert space $\mathcal{H}_\Lambda = \bigotimes\limits_{x\in\Lambda}L^2(\R,\mathrm{d}q_x)$ is given by the \emph{Schr\"odinger representation}
\begin{equation}\label{eq:Weyl_schroedinger_repr}
	W(f) = \exp\left(i\sum\limits_{x\in\Lambda}\left(\mathrm{Re}[f_x]q_x + \mathrm{Im}[f_x]p_x\right)\right)
		\qquad \textnormal{for all }f\in\ell^2(\Lambda),
\end{equation}
with the associated symplectic form $\sigma(f,g) = \mathrm{Im}\langle f,g\rangle$. An explicit computation shows that for the ground state density matrix $\rho_\Lambda$ of $H_\Lambda$ and any $f\in\ell^2(\Lambda)$ the action of the associated state on $W(f)$ reads
\begin{equation}\label{eq:gaussian_state}
	\omega_{\rho_\Lambda}(W(f)) := \Tr{\rho_\Lambda W(f)}
		= \exp\left(-\frac{1}{4}\langle \tilde{f},\Gamma_\Lambda\tilde{f}\rangle\right),
\end{equation}
where $\tilde{f}=(\mathrm{Re}[f],\mathrm{Im}[f])^T\in\R^{|\Lambda|}\oplus\R^{|\Lambda|}$ and $\Gamma_\Lambda$ denotes the covariance matrix associated to $\rho_\Lambda$, cf.~definition \eqref{eq:cov_matrix}. In particular, $\Gamma_\Lambda$ is here positive definite, cf.~\eqref{eq:m}.

\subsection{The entanglement entropy of Gaussian states}
\label{app:derivation}

We begin by giving a proof of Proposition~\ref{prop:reduction_Gaussian}, according to which the reduction of a Gaussian state remains Gaussian.

\begin{proof}[Proof of \cref{prop:reduction_Gaussian}]
	Consider the bipartition $\Lambda = \Lambda_0\cup\Lambda_0^c$. Denoting by $f_A$, $A\subset\Lambda$, the restriction of $f\in\ell^2(\Lambda)$ to $\ell^2(A)$ we have
	\begin{equation*}
		\omega_{\rho_\Lambda}(W(f)) = \Tr{\rho_\Lambda W(f_{\Lambda_0})\otimes W(f_{\Lambda_0^c}))}
	\end{equation*}
	Let $\rho_{\Lambda_0}=\trace_{\Lambda_0^c}(\rho_\Lambda)$ be the reduced density matrix on $\ell^2(\Lambda_0)$. With the definition of the partial trace we obtain
	\begin{align*}
		\omega_{\rho_{\Lambda_0}}(W(f_{\Lambda_0}))
			&:=\, \trace\left(\rho_{\Lambda_0}W(f_0)\right)
			= \Tr{\rho_\Lambda W(f_{\Lambda_0})\otimes\mathds{1}_{\Lambda_0^c}}
			= \Tr{\rho_\Lambda W(f_{\Lambda_0})\otimes W(0_{\Lambda_0^c})}\\
			&=\, \exp\left(-\frac{1}{4}\langle \tilde{f}_{\Lambda_0}\oplus 0_{\Lambda_0^c},
				\Gamma_\Lambda\tilde{f}_{\Lambda_0}\oplus 0_{\Lambda_0^c}\rangle\right)
			= \exp\left(-\frac{1}{4}\langle \tilde{f}_{\Lambda_0},
				\Gamma_{\Lambda_0}\tilde{f}_{\Lambda_0}\rangle\right),
	\end{align*}
	where we used the relation~\eqref{eq:gaussian_state}. Hence, the reduced state is still Gaussian with truncated covariance matrix $\Gamma_{\Lambda_0}$. The statement on the uncertainty relation is immediate.
\end{proof}

By Proposition~\ref{prop:entropy_formula} the entanglement entropy of Gaussian states can be computed explicitly in terms of the associated symplectic eigenvalues. Notwithstanding its ubiquity in the literature, we could not find a rigorous derivation thereof. For the reader's convenience, we spell out the argument hereafter.

\begin{proof}[Proof of \cref{prop:entropy_formula}]
	Consider a Gaussian state $\omega_\rho$ on $\mathcal{H}_\Lambda$ given by a density matrix $\rho$ with (positive definite) covariance matrix $\Gamma\in\R^{2|\Lambda|\times 2|\Lambda|}$. It is convenient to work in the Schr\"odinger representation of the Weyl algebra $\mathcal{W}\bigl(\ell^2(\Lambda)\bigr)$ introduced in \eqref{eq:Weyl_schroedinger_repr}. By Proposition~\ref{prop:williamson}, $\Gamma$ may be diagonalized by means of a symplectic matrix $S\in\mathsf{SP}(2|\Lambda|,\R)$. Let $U_S$ be the unitary implementation of $S$ in the representation of the Weyl algebra we are using, i.e.
	\begin{equation*}
		U_S W(f) U_S^{*} = W(S\tilde{f})\qquad \textnormal{for all }f\in\ell^2(\Lambda)
	\end{equation*}
	and $\tilde{f}=(\mathrm{Re}(f),\mathrm{Im}(f))^T$, where we identify $W(\tilde{f})\equiv W(f)$.

	The unitarily transformed density matrix $\rho_S := U_S^{*}\rho U_S$ has the same entanglement entropy as $\rho$ and the associated state $\omega_{\rho_S}$ is Gaussian with
	\begin{align*}
		\omega_{\rho_S}(W(f)) = \Tr{\rho W(S\tilde{f})}
			= \exp\left(-\frac{1}{4}\sum\limits_{x\in\Lambda}\gamma_x|f_x|^2\right)
			\qquad\textnormal{for all }f\in\ell^{2}(\Lambda),
	\end{align*}
	where $\lbrace\gamma_x\rbrace_{x\in\Lambda}$ are the symplectic eigenvalues of $\Gamma$.

	We construct an explicit density matrix $\sigma$ such that $\omega_{\rho_S}$ and $\omega_\sigma$ coincide on $\mathcal{W}\bigl(\ell^2(\Lambda)\bigr)$, whence by density $\sigma = \rho_S$, cf.~Lemma 3.1 in \cite{Nachtergaele2013}. Since the von Neumann entropy of $\sigma$ is then shown to satisfy \eqref{eq:von_neumann_entropy}, so does $\rho_S$ and in turn $\rho$. For this purpose let $\{\varphi_k\}_{k\in\N_0}\subset L^2(\R)$ be the set of orthonormal eigenfunctions of the one-dimensional harmonic oscillator
	\begin{align*}
		\varphi_k(x)=\frac{e^{-x^2/2}}{2^{k/2}\sqrt{k!}\pi^{1/4}}H_k(x),
	\end{align*}
	where $H_k(x)=(-1)^n e^{x^2}\frac{d^k}{dx^k}e^{-x^2}$ denotes the $k$-th Hermite polynomial. For any $z\in\C$ and $W(z):=\exp\left(i\left[\Re(z)q+\Im(z)p\right]\right)$ a Weyl operator in the Schr\"odinger representation of $\mathcal{W}(\C)$, a straightforward computation yields (cf. Theorem 3.1 in \cite{Abdul-Rahman2018})
	\begin{align*}
		\Braket{\varphi_k,W(z)\varphi_k}=e^{-|z|^2/4}L_k\left(\frac{|z|^2}{2}\right),
	\end{align*}
	where $L_k$ denotes the $k$-th Laguerre polynomial. Through the generating function
	\begin{align*}
		\sum_{k=0}^{\infty}t^kL_k(x)=\frac{1}{1-t}\exp\left(-\frac{tx}{1-t}\right)\qquad \textnormal{for }|t|<1
			\textnormal{ and }x\geq 0,
	\end{align*} 
	this implies by setting $t=(\gamma-1)/(\gamma+1)$
	\begin{align*}
		\frac{2}{\gamma+1}\sum_{k=0}^{\infty}\left(\frac{\gamma-1}{\gamma+1}\right)^k\Braket{\varphi_k,W(z)\varphi_k}
			=\exp\left(-\frac{\gamma |z|^2}{4}\right).
	\end{align*}
	The above equality can in fact be read
	\begin{align}\label{eq:weyl_mode}
		\Tr{\eta W(z)}=\exp\left(-\frac{\gamma |z|^2}{4}\right)
			\qquad \textnormal{for }\quad
			\eta:=\frac{2}{\gamma+1}\sum_{k=0}^{\infty}\left(\frac{\gamma-1}{\gamma+1}\right)^k
			\varphi_k\langle\varphi_k,\cdot\,\rangle,
	\end{align}
	where $\eta$ is seen to be a density operator by the geometric series. In particular, its von Neumann entropy may be computed as
	\begin{align*}
		S(\eta)&=\,-\Tr{\eta\log\eta}
			=-\frac{2}{\gamma+1}\sum_{k=0}^{\infty}\left(\frac{\gamma-1}{\gamma+1}\right)^k\log\left(\frac{2}{\gamma+1}	
				\left(\frac{\gamma-1}{\gamma+1}\right)^k\right)\\
			&=\,\frac{\gamma+1}{2}\log\left(\frac{\gamma+1}{2}\right)
				-\frac{\gamma-1}{2}\log\left(\frac{\gamma-1}{2}\right).
	\end{align*}
	We embed the orthonormal basis $\{\varphi_k\}_{k\in\N_0}$ of $L^2(\R,\mathrm{d}q)$ into $\mathcal{H}_\Lambda$ by defining
	\begin{align*}
		\varphi^{(x_i)}_k(x_1,\dots,x_{|\Lambda|}):=\varphi_k(x_i)
			\qquad\textnormal{for }i=1,\dots,|\Lambda|\textnormal{ and }k\in\N_0.
	\end{align*}
	Clearly, $\bigcup\limits_{x\in\Lambda}\{\varphi^{(x)}_k\}_{k\in\N_0}$ is an orthonormal basis of $\hilb_{\Lambda}$. We now define
	\begin{align*}
		\eta^{(x)}:=\frac{2}{\gamma_x+1}\sum_{k=0}^{\infty}\left(\frac{\gamma_x-1}{\gamma_x+1}\right)^k 
			\varphi_k^{(x)}\langle\varphi_k^{(x)},\cdot\,\rangle\qquad \textnormal{for }x\in\Lambda,
	\end{align*}
	where $\lbrace\gamma_x\rbrace_{x\in\Lambda}$ are the symplectic eigenvalues of $\Gamma$. Then \eqref{eq:weyl_mode} implies for $\sigma:=\bigotimes_{x\in\Lambda}\eta^{(x)}$ and all $f\in\ell^{2}(\Lambda)$
	\begin{align*}
		\omega_{\sigma}(W(f))=\exp\left(-\frac{1}{4}\sum_{x\in\Lambda}\gamma_x |f_x|^2\right)=\omega_{\rho_S}(W(f)).
	\end{align*}
	Formula \eqref{eq:von_neumann_entropy} follows by the additivity of the von Neumann entropy for product states.
\end{proof}

\section{Matrix elements of $h_\N^{[1/4]}$ and $h_\N^{[-1/2]}$}

We derive explicit expressions for the matrix elements of $h_\N^{[1/4]}$ and $h_\N^{[-1/2]}$ given as integrals in~\eqref{eq:hamiltonian_N}. These are in particular exploited to show the properties stated in Lemma~\ref{lem:properties_RS}.
\begin{lemma}\label{lem:matrix_elements_N}
 	For all $j,k\in\{1,\dots,n\}$, we have
	\begin{enumerate}[label=(\roman*)]
		\item\label{it:matel_1} $\displaystyle \left(h_\N^{[1/4]}\right)_{jk}=\frac{1}{2\sqrt{2\pi}}\left(\frac{\Gamma\left(j+k-\frac{1}{4}\right)}{\Gamma\left(j+k+\frac{5}{4}\right)}-\frac{\Gamma\left(|j-k|-\frac{1}{4}\right)}{\Gamma\left(|j-k|+\frac{5}{4}\right)}\right)$,
		\item\label{it:matel_2} $\displaystyle \left(h_\N^{[-1/2]}\right)_{jk}=\frac{1}{\pi}\left[\psi\left(j+k+\tfrac{1}{2}\right)-\psi\left(|j-k|+\tfrac{1}{2}\right)\right]=\frac{1}{\pi}\sum_{l=|j-k|+1}^{j+k}\frac{2}{2l-1}$,
	\end{enumerate}
	where $\Gamma(\cdot)$ denotes the Gamma function and  $\displaystyle \psi(\cdot)=\frac{\Gamma^\prime(\cdot)}{\Gamma(\cdot)}$ the digamma function.
\end{lemma}
\begin{proof}
	\ref{it:matel_1} From \eqref{eq:hamiltonian_N} and the identity $2\sin(\alpha)\sin(\beta) = \cos(\alpha-\beta)-\cos(\alpha+\beta)$, we have
	\begin{align*}
		\left(h_\N^{[1/4]}\right)_{jk} = \frac{2\sqrt{2}}{\pi}\int_0^{\pi/2} \sqrt{\sin(x)}
			\bigl(\cos(2|j-k|x)-\cos(2(j+k)x)\bigr)\mathrm{d}x.
	\end{align*}
	The result then immediately follows from the more general identity (see for instance 332 9b) in \cite{Groebner1950}): for any $n\in\mathbb{N}_0$ and $q>-1$ we have
	\begin{equation*}
		\int_0^{\pi/2}\sin(x)^q\cos(2lx)\,\mathrm{d}x 
			= (-1)^l\frac{\pi}{2^{q+1}}\frac{\Gamma(1+q)}{\Gamma(l+1+\tfrac{q}{2})\Gamma(1-l+\tfrac{q}{2})}.
	\end{equation*}
	In fact, for $q=1/2$, the right-hand side reads
	\begin{align*}
		(-1)^l \frac{\pi}{2\sqrt{2}}\frac{\Gamma(\tfrac{3}{2})}{\Gamma(l+\tfrac{5}{4})\Gamma(\tfrac{5}{4}-l)}
			= \frac{\sqrt{\pi}}{4\sqrt{2}}\frac{(-1)^l\pi}{\Gamma(l+\tfrac{5}{4})\Gamma(\tfrac{5}{4}-l)}
			= -\frac{\sqrt{\pi}}{8}\frac{\Gamma(l-\tfrac{1}{4})}{\Gamma(l+\tfrac{5}{4})},
	\end{align*}
	where the last equality is by the identity $\Gamma(z)\Gamma(1-z) = \pi/\sin(\pi z)$ for $z\not\in\Z$.
	\medskip\\
	\noindent \ref{it:matel_2} Assume without loss that $j\leq k$. Using $2\sin(\alpha)\cos(\beta) = \sin(\alpha+\beta) + \sin(\alpha-\beta)$ we first observe that
	\begin{align*}
		\frac{\sin(j\pi x)\sin(k\pi x)}{\sin\left(\frac{\pi x}{2}\right)}
			&=2\sum_{l=0}^{j-1}\cos\left(\left(l+\tfrac{1}{2}\right)\pi x\right)\sin(k\pi x)\\
			&=\sum_{l=0}^{j-1}\left(\sin\left(\left(l+k+\tfrac{1}{2}\right)\pi x\right)
			-\sin\left(\left(l-k+\tfrac{1}{2}\right)\pi x\right)\right).
	\end{align*}
	Inserting this into the integral expression~\eqref{eq:hamiltonian_N} we obtain
	\begin{align*}
		\left(h_\N^{[-1/2]}\right)_{jk}
			&=\,\sum_{l=0}^{j-1}\int_{0}^{1}\bigl(\sin\left(\left(l+k+\tfrac{1}{2}\right)\pi x\right)
				-\sin\left(\left(l-k+\tfrac{1}{2}\right)\pi x\right)\bigr)\mathrm{d}x\\
			&=\,\frac{1}{\pi}\sum_{l=0}^{j-1}\left(\frac{1}{l+k+\tfrac{1}{2}}-\frac{1}{l-k+\tfrac{1}{2}}\right)
			= \frac{1}{\pi}\sum_{l=|j-k|+1}^{j+k}\frac{2}{2l-1}.
	\end{align*}
	The expression with the digamma function follows from the identity
	\begin{align*}
		\psi\left(m+\tfrac{1}{2}\right)=-\gamma-2\log(2)+\sum_{l=1}^{m}\frac{2}{2l-1},\qquad m\in\N,
	\end{align*}
	where $\gamma$ denotes the Euler--Mascheroni constant.
\end{proof}

\bigskip
\noindent\textbf{Acknowledgments.} The research of V.B.~was partially supported by the Swiss National Science Foundation (Grant No.~P2EZP2-162235). J.S. acknowledges partial support by the German Academic Merit Foundation, the Max Weber program of the State of Bavaria and the TopMath program within the Elite Network of Bavaria. S.W. thanks the Simons foundation for support and the Centre de Recherches Math\'ematiques for hospitality during a long-term stay at the University of Montreal. All supports are gratefully acknowledged.
\bibliography{entropy_ocillators}
\bibliographystyle{abbrv}
\vspace{1cm}
\begin{center}
\begin{tabular}{m{4cm} m{4cm} m{4cm}}
\textsc{Vincent Beaud} & \textsc{Julian Sieber} & \textsc{Simone Warzel}\\
beaud@ma.tum.de & sieber@ma.tum.de & warzel@ma.tum.de
\end{tabular}
\end{center}

\end{document}